\documentclass[11pt]{article}

\usepackage{times}
\usepackage{color}
\usepackage{fullpage}
\usepackage{mathtools}
\usepackage{amsthm, amssymb}
\usepackage{xfrac}
\usepackage[ruled, vlined, linesnumbered]{algorithm2e}
\usepackage[colorlinks,linkcolor=blue,citecolor=blue,urlcolor=black]{hyperref}
\usepackage{caption, subcaption}
\DeclareCaptionType{copyrightbox}
\usepackage{framed,url}
\usepackage{stackrel}

\newtheorem{theorem}{Theorem}
\newtheorem{lemma}[theorem]{Lemma}

\newtheorem{claim}[theorem]{Claim}
\newtheorem{corollary}[theorem]{Corollary}

\newtheorem{observation}[theorem]{Observation}

\newtheorem{remark}[theorem]{Remark}

\newtheorem{definition}{Definition}

\DeclareMathOperator{\bE}{{\mathop{\mathbb{E}}}}

\newcommand{\eps}{\epsilon}
\newcommand{\E}{\mathcal{E}}

\newcommand{\abs}[1]{\left| #1 \right|}

\newcommand{\I}{\mathcal{I}}
\newcommand{\A}{\mathcal{A}}
\newcommand{\B}{\mathcal{B}}

\newcommand{\D}{\mathcal{D}}

\newcommand{\Acc}{\mathit{Acc}}
\newcommand{\Rej}{\mathit{Rej}}
\newcommand{\Top}{\mathit{Top}}

\newcommand{\cenT}{{\tt{CentralApproxTop}}}
\newcommand{\cenB}{{\tt{CentralApproxBtm}}}
\newcommand{\colM}{{\tt{CollabTopM}}}
\newcommand{\colMG}{{\tt{CollabTopMGeneral}}}
\newcommand{\colMI}{{\tt{CollabTopMImproved}}}

\newcommand{\colN}{{\tt{CollabTopMSimple}}}
\newcommand{\verM}{{\tt{VerifyTopM}}}
\newcommand{\lucb}{{\tt{LUCB}}}
\newcommand{\reduct}{{\tt{Reduction}}}
\newcommand{\reductG}{{\tt{ReductionGeneral}}}

\renewcommand{\ln}{\log}

\title{Collaborative Top Distribution Identifications with Limited Interaction\thanks{Accepted for presentation at the 61st Annual IEEE Symposium on Foundations of Computer Science (FOCS 2020). Nikolai Karpov and Qin Zhang are supported in part by NSF IIS-1633215, CCF-1844234 and CCF-2006591. Yuan Zhou is supported in part by NSF CCF-2006526 and a JPMorgan Chase AI Research Faculty Research Award.}}

\author{
Nikolai Karpov \\ 
Computer Science Department\\
Indiana University\\
\texttt{nkarpov@iu.edu}
\and 
Qin Zhang \\
Computer Science Department\\
Indiana University\\
\texttt{qzhangcs@indiana.edu}
\and 
Yuan Zhou\\
Department of ISE\\
University of Illinois at Urbana-Champaign\\
 \texttt{yuanz@illinois.edu}
}

\begin{document}

\maketitle

\begin{abstract}
We consider the following problem in this paper: given a set of $n$ distributions, find the top-$m$ ones with the largest means.  This problem is also called {\em top-$m$ arm identifications} in the literature of reinforcement learning, and has numerous applications. We study the problem in the collaborative learning model where we have multiple agents who can draw samples from the $n$ distributions in parallel.  Our goal is to characterize the tradeoffs between the running time of learning process and the number of rounds of interaction between agents, which is very expensive in various scenarios.  
We give optimal time-round tradeoffs, as well as demonstrate complexity separations between top-$1$ arm identification and top-$m$ arm identifications for general $m$ and between fixed-time and fixed-confidence variants.  As a byproduct, we also give an algorithm for selecting the distribution with the $m$-th largest mean in the collaborative learning model. 
\end{abstract}

\thispagestyle{empty}
\setcounter{page}{0}
\newpage

\section{Introduction}
\label{sec:intro}

In this paper we study the following problem: given a set of $n$ distributions, try to find the $m$ ones with the largest means via sampling.   We study the problem in the multi-agent setting where we have $K$ agents, who try to identify the top-$m$ distributions collaboratively via communication.  Suppose sampling from each distribution takes a unit time, our goal is to minimize both the running time and the number of rounds of communication of the collaborative learning process. 

The problem of top-$m$ distribution identifications originates from the literature of {\em multi-armed bandits} (MAB)~\cite{SB98}, where each distribution is called an {\em arm}, and each sampling from a distribution is called an {\em arm pull}.  When $m=1$, the problem is called {\em best arm identification}, and has been studied extensively in the centralized setting where there is only one agent~\cite{ABM10, BMS09, DGW02, EMM02, MT04, KKS13, JMNB14, KCG16, CLQ17a, CL16, GK16}.  Some of these algorithms can be easily modified to handle top-$m$ arm identification (e.g., \cite{ABM10,BWV13}).  
The problem of best arm identification has also been studied in the multi-agent collaborative learning model~\cite{HKKLS13,TZZ19}.  Surprisingly, we found that in the multi-agent setting, the tasks of identifying the best arm and the top-$m$ arms look to be very different in terms of problem complexities; the algorithm design and lower bound proof for the top-$m$ case require significantly new ideas, and need to address some fundamental challenges in collaborative learning.

\paragraph{Collaborative Learning with Limited Interaction.}  
A natural way to speed up machine learning tasks is to introduce multiple agents, and let them learn the target function collaboratively. In recent years some works have been done to address the power of parallelism  (under the name of {\em concurrent learning}, e.g., \cite{SNBWM13,GB15,DR18,DOR18}). Most of these works assume that agents have the full ability of communication. That is, they can send/receive messages to/from each other at {\em any} time step. This assumption, unfortunately, is unrealistic in real-world applications, as it would be very expensive to implement unrestricted communication, which is usually the biggest drain of time, data, energy and network bandwidth. For example, once we deploy sensors/robots to unknown environment such as deep sea and outer space, it would be almost impossible to recharge them; when we train a model in a central server by interacting with hundreds of thousands of mobile devices, the communication cost will directly contribute to our data bills, not mentioning the excessive energy and bandwidth consumption.

In this paper we consider the model of {\em collaborative learning with limited interaction}, where the learning process is partitioned into rounds of predefined time intervals. In each round, each of the $K$ agents takes a series of actions individually like in the centralized model, and they can only communicate at the end of each round. At the end of the last round before any communication, all agents should agree on the same output; otherwise we say the algorithm fails. Our goal is to minimize both the number of rounds of computation $R$ and the running time $T$ (assuming each action takes a unit time step).\footnote{We note that our model is a simplified version of the one formulated in \cite{TZZ19}. The model defined in \cite{TZZ19} allows each agent to perform different numbers of actions in each round, and the length of each round can be determined adaptively by the agents. However, we noticed that all the existing algorithms for collaborative learning in the literature have predefined round lengths,  under which there is no point for an agent to stop early in a round.}  

Naturally, there is a tradeoff between $R$ and $T$: If $R = 1$, that is, no communication is allowed, then $T \ge T_{\tt C}$ where $T_{\tt C}$ is the running time of the best centralized algorithm. When $R$ increases, $T$ may decrease. On the other hand we always have $T \ge T_{\tt C}/K$ even when $R = T$.  We are mostly interested in understanding the number of rounds needed to achieve almost {\em full  speedup}, that is, when $T = \tilde{O}(T_{\tt C}/K)$ where $\tilde{O}(\cdot)$ hides logarithmic factors.  

We do not put any constraints on the lengths of the messages that each agent can send at the end of each round, but in the MAB setting they will not be very large -- the information that each agent collects can always be compressed to an array of $n$ pairs in the form of $(x_i, \tilde{\theta}_i)$, where $x_i$ is the number of arm pulls on the $i$-th arm, and $\tilde{\theta}_i$ is the empirical mean of the $x_i$ arm pull.


\paragraph{Top-$m$ Arm Identification.}  
To be consistent with the MAB literature, we will use the term {\em arm} instead of {\em distribution} throughout this paper.
The top-$m$ arm identification problem is motivated by a variety of applications ranging from industrial engineering~\cite{KL85} to medical tests~\cite{Thompson33}, and from evolutionary computation~\cite{SBC06} to crowdsourcing~\cite{AAKS13}. The readers may refer to~\cite{ABM10, KTAS12, CLKLC14, CGL16, CLQ17} and references therein for the state-of-the-art results on the top-$m$ arm identification in the centralized model. 

In this paper we mainly focus on the {\em fixed-time} case, where given a fixed time horizon $T$, the task is to identify the set of $m$ arms with the largest means with the smallest error probability.  We will also discuss the {\em fixed-confidence} case, where given a fixed error probability $\delta$, the task is to identify the top-$m$ arms with error $\delta$ using the smallest amount of time. 

Without loss of generality, we assume that each of the underlying distributions has support on $(0,1)$. 
In the centralized setting, Bubeck et al.~\cite{BWV13} introduced the following complexity to characterize the hardness of an input instance $V$ for the top-$m$ arm identification problem. Let $\theta_i$ be the mean of the $i$-th arm. Let $[j]$ be the index of the arm in $V$ with the $j$-th largest mean, and let $\theta_{[j]}(V)$ be the corresponding mean.   
Given an input instance $I$ of $n$ arms, let $\Delta^{\langle m\rangle}_i(I)$ be the gap between the mean of the $i$-th arm and that of the $[m]$-th arm or the $[m+1]$-th arm, whichever is larger. In other words,
\begin{equation}
\label{eq:gap}
    \Delta^{\langle m \rangle}_i(I) \triangleq
    \begin{cases}
    \theta_{i} - \theta_{[m + 1]}(I), & \text{if}~\theta_i \geq \theta_{[m]}(I), \\
    \theta_{[m]}(I) - \theta_{i}, & \text{if}~\theta_i \leq \theta_{[m + 1]}(I).
    \end{cases}
\end{equation}

\begin{definition}[{\bf Instance Complexity}]
\label{def:instance-complexity}

Given an input instance $I$ of $n$ arms and a parameter $m$ (call it the {\em pivot}), we define the following quantity which characterizes the complexity of $I$.
\begin{equation*}
H^{\langle m \rangle}(I) \triangleq \sum\limits_{i \in I} {\left(\Delta^{\langle m\rangle}_i(I)\right)^{-2}}.
\end{equation*}
We also define a related quantity which we call the {\em $\eps$-truncated instance complexity}. 
\begin{equation*}
H^{\langle m \rangle}_{\epsilon}(I) \triangleq \sum\limits_{i \in I} {\max\left\{\Delta^{\langle m\rangle}_i(I), \epsilon\right\}^{-2}}.
\end{equation*}
\end{definition}

To see why $H^{\langle m \rangle}(I)$ is the right measure for the instance complexity, note that if the mean of an arm is either $(\theta+\Delta)$ or $(\theta-\Delta)$ where $\theta$ is a known threshold, it takes $\Omega(\Delta^{-2})$ samples to decide whether the mean is above or below the threshold $\theta$ (as long as $\theta\pm\Delta$ are bounded away from $0$ and $1$). Therefore, suppose all the means are bounded away from $0$ and $1$, even if we are given the means of the $[m]$-th and the $[m+1]$-th arms, it still takes  $\Omega(H^{\langle m \rangle}(I))$ samples to decide for each arm whether it is one of the top-$m$ arms or not. Such intuition can be formalized to show that, in the fixed-confidence case, $\Omega(H^{\langle m \rangle}(I) \log (1/\delta))$ samples are needed to identify the top-$m$ arms with success probability $(1-\delta)$~\cite{SJR17,CLQ17}. On the other hand, there are centralized algorithms to achieve $O(H^{\langle m \rangle}(I) \log (1/\delta)+H^{\langle m \rangle}(I) \log H^{\langle m \rangle}(I))$ (see, e.g., \cite{KTAS12}), almost matching the lower bound (up to logarithmic factors).

For the fixed-time case, in \cite{BWV13} it was shown that there is a centralized algorithm that identifies the top-$m$ arms with probability at least
\begin{equation}
\label{eq:cen-ub}
1 - \exp\left(- \tilde{\Omega}\left(\frac{T}{H^{\langle m \rangle}(I)} \right)\right)
\end{equation}
using at most $T$ time steps, where $\tilde{\Omega}(\cdot)$ hides logarithmic factors in $n$.  This upper bound can also be shown to be tight up to logarithmic factors~\cite{KCG16,CL16,SJR17,CLQ17}.  In the collaborative learning setting, our goal is to replace the $T$ factor in (\ref{eq:cen-ub}) with $KT$ where $K$ is the number of agents, so as to achieve a full speedup.

\paragraph{Our Contributions.}  We summarize our main results and their implications.
\begin{enumerate}
\item  
We give an algorithm for the fixed-time top-$m$ arm identification problem in the collaborative learning model with $K$ agents and a set $I$ of $n$ arms. For any choice of $r$, the algorithm uses $T$ time steps and $O(\log\frac{\log m}{\log K} + r)$ rounds of communication, and successfully computes the set of top-$m$ arms with probability at least 
$1 - \exp\left(-\tilde{\Omega}\left(\frac{K^{(R - 1) / R} \cdot T}{H^{\langle m \rangle}(I)}\right)\right)$. In particular, when $r = \log K$, the algorithm uses $T$ time steps and $O(\log \frac{\log m}{\log K}+ \log K)$ rounds of communication to compute the set of top-$m$ arms with probability at least $1 - \exp\left(-\tilde{\Omega}\left(\frac{KT}{H^{\langle m \rangle}(I)}\right)\right)$, achieving a full speedup.
See Section~\ref{sec:ub}.

\item  We prove that under the same setting, any collaborative algorithm that uses $T = \frac{1}{\sqrt{K}} \cdot H^{\langle m \rangle}(I)$ time steps and aims to achieve success probability $0.99$ needs at least $\Omega(\log\frac{\log m}{\log K})$ rounds of communication.  By leveraging a result in \cite{TZZ19}, we can also show that  any collaborative algorithm that uses $T = \frac{\alpha}{{K}} \cdot H^{\langle m \rangle}(I)$ time steps and aims to achieve success probability $0.99$ needs at least $\Omega({\log K}/{(\log\log K + \log \alpha)})$ rounds of communication. These indicate that our upper bound is almost the best possible.  See Section~\ref{sec:lb}.

\item Our lower bound gives a strong separation between the best arm identification and top-$m$ identifications: there is a collaborative algorithm for  best arm identification (i.e., when $m=1$) that uses $T = \tilde{O}\left(\frac{1}{\sqrt{K}} \cdot H^{\langle 1 \rangle}(I)\right)$ time and $2$ rounds of communication (see \cite{TZZ19,HKKLS13}), while Item 2 states that for general $m$, to achieve the same time bound we need $\Omega({\log K}/{(\log\log K + \log \alpha)})$ rounds of communication.

\item We give an algorithm for the fixed-confidence top-$m$ identification problem in the collaborative model with $K$ agents and a set of $n$ arms; the algorithm uses $O\left(\frac{H^{\langle m\rangle}(I)}{K} \log\left(\frac{n}{\delta} \log{H^{\langle m\rangle}}\right)\right)$ time steps and $O\left(\log (1/\Delta^{\langle m\rangle}_{[m]})\right)$ rounds of communication, and successfully computes the set of top-$m$ arms with probability at least $1 - \delta$. This is almost tight by a previous result in \cite{TZZ19}. See Section~\ref{sec:confidence}. 

\item Combining Items 1, 2, and 4, we have given a separation between fixed-time and fixed-confidence top-$m$ arm identification. We note that a similar separation result is also proved for the best arm identification problem~\cite{TZZ19}, although the round complexities for top-$m$ identification are quite different from the $m=1$ special case (i.e., best arm identification).
\end{enumerate}

\paragraph{Speedup.}
In \cite{TZZ19} the authors introduced a concept called {\em speedup} for presenting the power of collaborative learning algorithms.  The precise definition of speedup is rather complicated due to the definition of the instance complexity of MAB.  Roughly, the speedup is defined to be the ratio between the {\em best} running time of centralized algorithm and that of a collaborative algorithm (given a predefined round budget $R$) under the condition that the two algorithms achieve the same success probability. In this paper we simply focus on a fixed success probability $0.99$, and define the speedup of a collaborative algorithm which identifies the top-$m$ arms on input instance $I$ with accuracy $0.99$ using time $T_{\A}(I)$ to be $T_{\A}(I) / H^{\langle m\rangle}(I)$, since the best centralized algorithm achieving success probability $0.99$ has running time $\tilde{\Theta}(H^{\langle m\rangle}(I))$ \cite{BWV13}. Interpreting our results in terms of speedup, we have the following remarks:
\begin{enumerate}
\item Our algorithm for fixed-time top-$m$ arm identification achieves a speedup of $\tilde{O}(K^{\frac{r-1}{r}})$ and uses $O(\log\frac{\log m}{\log K} + r)$ rounds.  

\item Our lower bound shows that in order to achieve even an $\tilde{\Omega}(\sqrt{K})$ speedup, any algorithm for top-$m$ arm identification needs at least $\Omega(\log\frac{\log m}{\log K})$ rounds.  

\item Compared with the main result for the best arm identification in \cite{TZZ19}, which states that there is a $R$-round algorithm achieving a speedup of $\tilde{O}(K^{\frac{R-1}{R}})$, we have shown a separation between the complexities of the two problems (e.g., when $R = 2$).
\end{enumerate}

\paragraph{Selection under Uncertainty.}  As a byproduct,
we also get almost tight bounds for a closely related problem we call  {\em selection under uncertainty}.  This problem is similar to the classic {\em selection} problem where given a set of $n$ numbers, one needs to find the $m$-th largest number.  The difference is that now instead of having $n$ (deterministic) numbers, we have $n$ distributions/arms, and our goal is to find the one with the $m$-th largest mean via sampling.  It is easy to see that this problem can be solved by first identifying the top-$m$ arms, and then finding the worst arm in these top-$m$ arms, which can be done in the same way as identifying the best arm.

For convenience, let us introduce a new (but very similar) definition of instance complexity for the selection under uncertainty problem:
\[
\bar{H}^{\langle m\rangle}(I) \triangleq \sum\limits_{i \neq [m]} (\theta_i - \theta_{[m]})^{-2}.
\]
With $\bar{H}^{\langle m\rangle}$ we have the following immediate result:
\begin{itemize}
\item There exists an algorithm for the fixed-time $m$-th arm selection problem in the collaborative learning model with $K$ agents and a set $I$ of $n$ arms; the algorithm uses $T$ time steps and $O(\log\frac{\log m}{\log K} + r)$ rounds of communication, and successfully identifies the $m$-th arm with probability at least $1 -  \exp\left(-\tilde{\Omega}\left(\frac{K^{(r - 1) / r} \cdot T}{\bar{H}^{\langle m \rangle}(I)}\right)\right)$.
\end{itemize}

\paragraph{Why Top-$m$ Arm Identification is Difficult in the Collaborative Learning Model?}
Before presenting our results, let us first try to give some intuition on why top-$m$ arm identification is difficult in the collaborative learning setting, as one may think that the top-$m$ arm identification is a natural generalization of best arm identification (when $m=1$), and the algorithm for the latter in \cite{TZZ19} may be adapted to the former.  

The key procedure used in previous collaborative algorithms for best arm identification~\cite{HKKLS13,TZZ19} is that in the first round, we {\em randomly} partition the set $I$ of $n$ arms into $K$ groups, and feed each group to one agent as a subproblem. Now if each of the $K$ agents computes the best arm in its subproblem, then we can reduce the number of best arm candidates from $n$ to $K$ after the first round, which is critical for us to achieve $\log K$ communication rounds.  The question now is whether each subproblem can be solved time-efficiently (more precisely, in $\tilde{O}(H^{\langle 1 \rangle}(I)/K)$ time steps if  we target a $\tilde{\Omega}(K)$ speedup) at each agent in the first round.

A nice property for the best arm identification is that if we randomly partition the set $I$ of $n$ arms to the $K$ groups, then the group (denoted by $G$) containing the global best arm has a subproblem complexity
$H' = \sum_{i =2}^{\abs{G}} \left(\Delta'_i\right)^{-2}$, where $\Delta'_i$ is the difference between the mean of the best arm and that of the $i$-th best arm in group $G$.  It is easy to show that 
\begin{equation}
\label{eq:partition}
\bE[H'] = \Theta\left(H^{\langle 1 \rangle}(I)/K\right).
\end{equation}  
Therefore, even though we cannot guarantee that each of the $K$ subproblems can be solved successfully under time budget $\tilde{O}(H/K)$, we still know that the global best arm will advance to the next round with a good probability, which is enough for the algorithm to succeed.

Unfortunately, the above property does not hold in the top-$m$ setting due to its ``multi-objective'' goal.  First, the global $m$-th arm will only be assigned to one agent, and thus others do not know what pivot to use for defining its subproblem complexity. Second, even for the agent who gets the $m$-th arm $j$, it does not know what is the local rank of $j$, and, thus, still does not know when to stop the local pruning. Third, even if the agents know the local ranks of the $m$-th arm, it may not have enough time budget to solve the sub-problem; note that this is an issue only for the top-$m$ case but {\em not} for the best arm case, since in the top-$m$ case each subproblem may contain some top-$m$ arms. 

We will design an algorithm which addresses all of these challenges, and then complement it with an almost tight lower bound.  Looking back, we feel that in the best arm case it was just lucky for us to have Equation~(\ref{eq:partition}), while in the general top-$m$ case we have to deal with some inherent challenges in collaborative learning, which, unfortunately, also make our algorithm for top-$m$ much more complicated than that for best arm identification.  We will give a technical overview for both the algorithm and lower bound proof in Section~\ref{sec:overview}.

\paragraph{Related Work.}
To the best of our knowledge, the collaborative learning model studied in this paper was first proposed in \cite{HKKLS13}, where the authors studied the best arm identification problem in MAB.  The model was recently formalized in \cite{TZZ19}, where almost tight time-round tradeoffs for best arm identification are given.

A number of works studied {\em regret minimization}, which is another important problem in MAB, in various distributed models, most of which are different from the collaborative learning model considered in this paper.  For example, several works \cite{LZ10,RSS16,BL18} studied regret minimization in the setting of cognitive ratio network, where radio channels are models as arms, and the rewards by pulling each arm depend on the number of simultaneous pulls by the $K$ agents (i.e., penalty is introduced for collisions).  In \cite{CCDJ17} the authors considered a model where at each time step each agent can choose either to pull an arm, or broadcast a message to other agents, but cannot do both. Authors of \cite{SBHOJK13,LSL16,XTZV15} considered regret minimization in communication networks. Distributed regret minimization has also been studied in the non-stochastic setting~\cite{AK05,KLR12,CGMM16}.

The collaborative learning model is closely related to the {\em batched} model (or, {\em learning with limited adaptivity}), where one wants to minimize the number of policy switches in the learning process.  In the batched model we want to minimize the number of policy switches when trying to achieve our learning goal.  Algorithms designed in the batched model can naturally be translated to a {\em restricted} version of the collaborative model in which at each time step, the action taken by each agent is determined by the information (historical actions and outcomes, messages received from other agents, and the randomness of the algorithm) the agent has at the beginning of the round, and the agents cannot change their policies in the middle of the a round.  A number of problems have been studied in the batched model in recent years, including best arm identification \cite{JJNZ16,AAAK17, JSXC19}, regret minimization in MAB~\cite{PRCS15,GHRZ19, EKMM19}, $Q$-learning~\cite{BXJW19}, convex optimization~\cite{DRY18}, online learning~\cite{CDS13}.  We note that our collaboratively learning algorithm for top-$m$ arm identification in the fixed-confidence case also works in the batched model, and improves the algorithm in \cite{JSXC19}.


Finally, we note that there is also a large body of work on sample/communication-efficient distributed algorithms for various learning-related tasks such as classification~\cite{BBFM12,DPSV12,KLMY19}, convex optimization~\cite{ZWSL10,ZDW12,AS15}, linear programming~\cite{AKZ19, VWW19}. Sample-efficient PAC learning in the collaborative setting is recently studied by~\cite{BHPQ17, CZZ18, NZ18}. However, the models considered in the papers mentioned above mainly focus on reducing the sample/communication cost, and are all different from the collaborative learning with limited interaction model we study in this paper.

\paragraph{Notations and Conventions.}
Let $\Top_m(V)$ be the indices of $m$ arms in $V$ with the largest means, and $\Top_1(V)$ be the index of the best arm in $V$.  

We say the $i$-th arm is $(\eps, j)$-top in $V$ if and only if $\theta_i \ge \theta_{[j]}(V) - \eps$. Similarly, the $i$-th arm is $(\eps, j)$-bottom in $V$ if and only if $\theta_i \le \theta_{[\abs{V}+1-j]}(V) + \eps$.

In this paper we focus on the case when $\theta_{[m]}(I) > \theta_{[m+1]}(I)$, since otherwise the instance complexity of $I$ will be infinity.

For simplicity, we will write $\Top_m(V)$, $\Top_1(V)$, $\theta_{[i]}(V)$, $\Delta_i^{\langle m \rangle}(V)$, $H^{\langle m \rangle}(V)$, and $H_{\eps}^{\langle m \rangle}(V)$ as $\Top_m$, $\Top_1$, $\theta_{[i]}$, $\Delta_i^{\langle m \rangle}$, $H^{\langle m \rangle}$, and $H_{\eps}^{\langle m \rangle}$, when $V = I$ ($I$ is the input instance) or it is clear from the context.

We include a list of frequently used (global) notations in Table~\ref{tab:notation}.

\begin{table}[t]
\centering
\begin{tabular}{|c|c|} 
\hline 
$n$ & number of arms in the input instance.\\
\hline 
$K$ & number of agents. \\
\hline
$T$ & running time. \\
\hline
$\theta_i$ & mean of the $i$-th arm.\\
\hline
$\theta_{[i]}(V)$ & the $i$-th largest mean among arms in $V$. \\
\hline
$\Top_m(V)$ & indices of the $m$ arms with the largest means in $V$.\\
\hline
$\Top_1(V)$ & index of the best arm in $V$.\\
\hline
$\Delta_i^{\langle m \rangle}(V)$ & mean gap of the $i$-th arm; defined in (\ref{eq:gap}).\\
\hline
$H^{\langle m\rangle}(V)$ & instance complexity; see Definition~\ref{def:instance-complexity}.\\
\hline
$H^{\langle m\rangle}_{\eps}(V)$ & $\eps$-truncated instance complexity; see Definition~\ref{def:instance-complexity}.\\
\hline
\end{tabular}
\caption{Summary of Notations}
\label{tab:notation}
\end{table}

\paragraph{Roadmap.} In the rest of the paper, we first give a technical overview of our main results in Section~\ref{sec:overview}.  We next present our algorithm for the fixed-time case in Section~\ref{sec:ub}, and then complement it with a matching lower bound in Section~\ref{sec:lb}. Finally in Section~\ref{sec:confidence}, we give an algorithm for the fixed-confidence case and discuss the correponding lower bound.

\section{Technical Overview} 
\label{sec:overview}

In this section we give a technical overview for our upper and lower bounds for fixed-time top-$m$ arm identification. 

\subsection{Upper Bounds for the Fixed-Time Setting}
\label{sec:ub-overview}

For simplicity we consider the full speedup setting (i.e., we target a speedup of $\tilde{\Omega}(K)$); the general speedup is an easy extension. We achieve our upper bound result for fixed-time top-$m$ arm identification in three stages.  We first design an algorithm for a special time horizon $T = \tilde{\Theta}(H^{\langle m \rangle}/K)$ which uses $O(\log\frac{\log n}{\log K} + \log K)$ rounds of communication and has an error probability $0.01$. We next consider general time horizon $T$, and target an error probability that is exponentially small in $T$.  Finally, we try to improve the round complexity to $O(\log\frac{\log m}{\log K} + \log K)$.  In each stage
we face new challenges which stem from the collaborative learning model, each of which demands novel ideas.

\paragraph{Stage $1$: A Basic Algorithm.}  We start with our basic algorithm.  A natural idea for achieving the $T = \tilde{O}(H^{\langle m \rangle}/K)$ running time is to {\em randomly} partition the $n$ arms to $K$ agents, and then ask each agent to solve a top-$\eta$ arms identification (for some value $\eta$) on its sub-instance.  At the end we try to aggregate the $K$ outputs.  As briefly mentioned in the introduction, there are multiple hurdles associated with this approach.  First, it is not clear how to set the value $\eta$, since we do not know how many global top-$m$ arms will be distributed to each agent.  Second, even if we know the number of global top-$m$ arms assigned to each agent, there are cases in which the global instance complexity is rarely distributed evenly across the $K$ agents. In other words, we cannot guarantee that each agent can solve the subproblem within our time budget $\tilde{O}(H^{\langle m \rangle}/K)$.

We resolve these issues using the following ideas: we take a {\em conservative} approach by setting $\eta \approx (m/K - \sqrt{n})$, and ask each agent to adopt a PAC algorithm for multiple arm identification and compute an {\em approximate} set of top-$\eta$ arms on its sub-instance using $\tilde{O}(H^{\langle m \rangle}/K)$ time steps.  The approximation error is a {\em random variable} depending on the random partition process. We then show that with a good probability this error is smaller than the gap between the smallest mean of the outputted arms and that of the {\em global} $m$-th arm. In this way we can guarantee that the approximate top-$\eta$ arms outputted by each agent are indeed in the set of global top-$m$ arms. Using the same idea we try to prune a set of ``bottom'' arms of size $\approx ((n-m)/K-\sqrt{n})$.  After these operations we recurse on the rest $O(K\sqrt{n})$ arms.  We continue the recursion for $O(\log\frac{\log n}{\log K})$ rounds until the number of arms is reduced to $K^{10}$, and then use a simple $O(\log n)$-round collaborative algorithm which is modified from an existing centralized algorithm. Note that for $n' = K^{10}$ we have $O(\log n') = O(\log K)$, and thus overall we have used $O(\log\frac{\log n}{\log K} + \log K)$ rounds.

\paragraph{Stage $2$: General Time Horizon.}
The basic algorithm only guarantees that the set of top-$m$ arms are correctly identified with probability $0.99$.  Our next goal is to make the error probability exponentially small in $T$, which is achievable in the centralized setting. The standard technique to achieve this is to perform parallel repetition and then take the majority. That is, we guess the instance complexity to be $H = 1, 2, 4, \ldots$, and for each guess we run the basic algorithm with time horizon $H$ for $T/H$ times.  Finally, we take the majority of the output.  In the case that the budget $T$ is larger than the actual instance complexity, at each run with probability $0.99$ we are guaranteed to obtain the correct answer.  Unfortunately, when $T$ is smaller than the actual instance complexity, not much can be guaranteed. For some bad input instances, the output of the basic algorithm can be {\em consistently} wrong, resulting in a wrong majority.  

We resolve this difficulty by introducing a notion we call {\em top-$m$ certificate}, which takes form of a pair $(S, \{\tilde{\theta_i}\}_{i \in I})$, with the property that $S = \Top_m$ and for each $i \in I$, it holds that $\abs{\tilde{\theta_i} - \theta_i} < \Delta_i^{\langle m \rangle}/4$. We can augment our basic algorithm to output a $(S, \{\tilde{\theta_i}\}_{i \in I})$ pair (instead of simply a set of top-$m$ arms). We then design a verification algorithm which is able to check for each $(S, \{\tilde{\theta_i}\}_{i \in I})$ pair whether it is indeed a top-$m$ certificate.  Our verification step can be fully parallelized and can finish within our guessed instance complexity $H$.  With such a verification step at hand, the situation that we take a wrong majority will {\em not} happen with high probability.

\paragraph{Stage $3$: Better Round Complexity.}
Our ultimate goal is to achieve an $O(\log\frac{\log m}{\log K} + \log K)$ round complexity, instead of $O(\log\frac{\log n}{\log K} + \log K)$ in the basic algorithm.  We approach this by first reducing the number of arms in the input instance to $\tilde{O}(m)$, and then applying the basic algorithm.  Such a reduction, however, is highly non-trivial, especially when we require the error probability introduced by the reduction to again be exponentially small in $T$.

Our basic idea for performing the reduction is the following: we construct a random sub-instance $V$ by sampling each of the $n$ arms with probability $1/m$. We can show that with constant probability, $V$ contains exactly one global top-$m$ arm, and $H^{\langle 1 \rangle}(V) = O(H^{\langle m \rangle}/m)$.  Therefore we have enough time budget to compute the best arm of $V$ and include it into set $S$ as a top-$m$ candidate.  We perform this subsampling procedure for $\tilde{O}(m)$ times, getting $\tilde{O}(m)$ sub-instances.  By the Coupon Collector's problem we know that all global top-$m$ arms will be included in $S$ with a good probability. 

The challenging part is to reduce the error probability of this reduction to a value that is exponentially small in $T$.  Unfortunate, the idea of ``guess-then-verify'' that we have used previously does not apply here -- there is simply no $(S, \{\tilde{\theta_i}\}_{i \in I})$ pair for us to verify in the reduction process.

We take the following new approach.  We try to make sure that for each randomly sampled sub-instance on which we try to compute the best arm, the probability of outputting any arm in $\Top_m$ is at least {\em half} of that of any arm outside $\Top_m$. This turns out to be enough for us to guarantee that the set $S$ contains all top-$m$ arms.
 We comment that the relaxation ``half'' is necessary here for a technical reason which we will elaborate next. 

Our key observation is that if we provide sufficient time budget, say, $T \ge \lambda H^{\langle 1 \rangle}(V)$ where $\lambda$ is a polylogarithmic factor, for solving a randomly sampled sub-instance $V$, then provided that there is only one arm $a \in \Top_m$ in $V$, we will output $a$ correctly with a good probability. Now for any two arms $a \in \Top_m$ and $b \not\in \Top_m$, by the uniformity of the sampling they will be in the sub-instance with equal probability.  We are thus able to conclude that the probability of outputting $a$ is at least as large as that of outputting $b$.   On the other hand, if $T \le H^{\langle 1 \rangle}(V)$, then we can use our verification step to detect this event.  The subtle part is the middle case when $H^{\langle 1 \rangle}(V) \le T \le \lambda H^{\langle 1 \rangle}(V)$, to handle which we perturb our time budget $T$ such that it takes values $T/\lambda$ or $\lambda T$ with equal probability. Using this trick we are able to ``reduce'' the third case to the first two cases with probability at least $1/2$, which leads to our desired property. The actual implementation of this idea is more involved, and we refer the readers to Section~\ref{sec:ub-improve} for details.

\subsection{Lower Bounds for the Fixed-Time Setting}
\label{sec:lb-overview}

In the lower bound part, we present two results. The first result is that $\Omega(\log K / (\log \log K + \log \alpha))$ communication rounds are needed for any algorithm with $(K / \alpha)$ speedup to identify the top-$m$ arms for any $m$. This matches (up to logarithmic factors) the $R$ term in the $O(\log \frac{\log m}{\log K} + R)$ rounds vs $\tilde{O}(K^{(R-1)/R})$ speedup trade-off in our upper bound result. This lower bound theorem is derived via a simple reduction together with the similar type of lower bound proved in \cite{TZZ19} for the $m=1$ special case.

Our main contribution in the lower bound part is the second theorem. The theorem states that even if the goal is an $O(\sqrt{K})$ speedup, the  $\log \frac{\log m}{\log K}$ term in the round-speedup trade-off is necessary. (In fact, the $\log \frac{\log m}{\log K}$ can be shown to be necessary for any $K^{\zeta}$ speedup where $\zeta$ is a positive constant.) This marks a completely different phenomenon from the $m=1$ special case where only $2$ rounds of communication are needed to achieve an $\tilde{O}(\sqrt{K})$ speedup \cite{TZZ19,HKKLS13}. Below we sketch the proof idea for this lower bound theorem.

The need for the $\log \frac{\log m}{\log K}$ term in the round complexity stems from the hardness of collaboratively learning the splitting position (i.e., where the $m$-th largest arm locates), which turns out to be substantially more difficult than estimating the best arm (the $m=1$ special case). We start from the fact that any (possibly randomized) algorithm cannot identify the number of $1$'s in the $n$-bit binary vector with success probability $\omega(n^{-1/2})$, if the algorithm is allowed to probe only $o(n)$ entries in the vector. A strengthened statement we will prove as the building block is the following lower bound for the ``learning the bias'' problem: given $n$ Bernoulli arms (i.e., the stochastic reward of the arm is either $0$ or $1$), each of which has mean reward $(\mu + \epsilon)$ or $(\mu - \epsilon)$, then any algorithm using $o(n \epsilon^{-2}/ \log (n/\epsilon))$ samples will not be able to identify the number of two types of arms with probability $\omega(n^{-1/2})$.\footnote{The sample complexity lower bound for a similar problem is proved in a recent work \cite{LV19}. Our lower bound is different from theirs in two aspects. First, in their setting, the number of arms is not bounded and the goal is to estimate the fraction of the two types of arms up to an additive error, while in our setting, the number of arms is $n$, and the goal is to find out the exact numbers of arms for the two types. Second, their lower bound is for algorithms with constant success probability, while our lower bound is for algorithms with only $\omega(n^{-1/2})$ success probability.}

Now we explain the connection between the learning the bias problem and the top-$m$ arm identification problem by sketching the plan of constructing the hard instances as follows. Suppose that we set all but $n^{1/2}$ arms in the hard instance to be Bernoulli with mean reward either $(\mu + \epsilon)$ (namely ``the top arms'') or $(\mu - \epsilon)$ (namely ``the bottom arms''). We denote the set of the rest $n^{1/2}$ arms by $M$, and their mean rewards are sandwiched between $(\mu + \epsilon)$ and $(\mu - \epsilon)$. We will set $m = n/2$, i.e., the goal is to identify the top half of the arms. Now, as long as the number of top arms, denoted by $X$, is bounded between $\frac{n}{2} - \sqrt{n}$ and $\frac{n}{2} + \sqrt{n}$, the goal is equivalent to identify the $X$ top arms and the top-$(\frac{n}{2} - X)$ arms in $M$. We then vary the number of the top arms and consequently the number of the bottom arms (say, let $X$ be uniformly randomly chosen from the range), and will argue that each agent will not be able to identify $X$ much better than a random guess without communication, and therefore must perform one round of communication to learn $X$ in order to identify the top-$(\frac{n}{2} - X)$ arms in $M$.  Here, the need for communication is due to the lower bound for learning the bias and the fact that any agent in a $\sqrt{K}$-speedup algorithm is allowed to make only $O(n \epsilon^{-2} / \sqrt{K}) = o(n \epsilon^{-2})$ samples (where we make a crucial assumption that the $H^{\langle m \rangle}$ complexity of the constructed hard instance is $O(n \epsilon^{-2})$). The last piece of plan is to argue that since $X$ is not known before the first round of communication, each agent cannot make much progress before the communication towards identifying the top-$(\frac{n}{2} - X)$ arms in $M$, which is a necessary sub-task. We will finally inductively prove a communication lower bound for this sub-task. Note that the number of arms in $M$ is $n^{1/2}$, and this plan will lead to a $\log \log n$-style round complexity lower bound. 

There are several challenges for the plan above. Note that in the sub-task for $M$, the goal is no longer to identify the top half arms, which is not well aligned with the (planned) induction hypothesis. Moreover, to make the induction work, $M$ would naturally have the similar structure as the $n$-arm instance, i.e., with many top and bottom arms (possibly with different $\mu$ and $\epsilon$ parameters). However, such a construction would hardly ensure that the $H^{\langle m' \rangle}$ complexity is still $O(n \epsilon^{-2})$. Indeed, if the goal of the sub-task is to identify, for example, the top $|M|/4$ arms, since most of the top half arms are the same, the corresponding the $H^{\langle m' \rangle}$ complexity would become infinitely large. Finally, it is not clear how to make sure that any agent will not gain much information about $M$ before the first round of communication so as to quickly identify the top $(\frac{n}{2} - X)$ arms in $M$ whenever $X$ is learned.

To address these challenges, we craft a more complex distribution of hierarchical instances. The main highlight is that we let $M$ consist of multiple blocks $I_1, I_2, \dots, I_k$, where each block has the same number of arms and is independently sampled from a recursively defined hard distribution. We restrict the possible values of $(\frac{n}{2} - X)$ to be the half multiples of the block size so that the sub-task always becomes to identify the top half arms in $I_{\xi}$ for some $\xi \in \{1, 2, \dots k\}$. We will make careful selection of the block parameters so that the $H^{\langle m \rangle}$ complexity for any instance in the support of the distribution, and the $H^{\langle m' \rangle}$ complexity of any sub-task, are all $\tilde{\Theta} (n\epsilon^{-2})$, where both upper and lower bounds are crucial to the proof.

The detailed construction and proof idea are presented in Section~\ref{sec:hard-instance}. For the analysis of the learning the bias problem and the ultimate round complexity lower bound, please refer to the subsequent subsections (Section~\ref{sec:learn-bias} and Section~\ref{lb:final-recursion}).

\section{A Collaborative Algorithm for the Fixed-Time Case}
\label{sec:ub}

\subsection{Preparation}
\label{sec:prepare}

In this section we give a few auxiliary lemmas to be used in our algorithm and analysis.

The following lemma gives connections between instance complexities and sub-instance complexities.
\begin{lemma}
\label{lem:sandwich}

Let $V \subseteq I\ (\abs{I} = n)$ be a subset of arms. Let $j \in \{1, \ldots, n-1\}$ and $k \in \left\{1, \dotsc, |V| - 1\right\}$ be two indices such that $\theta_{[k]}(V) \geq \theta_{[j]} \ge \theta_{[j + 1]} \geq \theta_{[k + 1]}(V)$. We have 
\begin{enumerate}
\item $H^{\langle k\rangle}(V) \leq \sum \limits_{i \in V} \left(\Delta^{\langle j \rangle}_i\right)^{-2} \leq H^{\langle j \rangle}$.

\item $H^{\langle k\rangle}_{\epsilon}(V) \leq \sum \limits_{i \in V} \max\{\Delta^{\langle j \rangle}_i, \epsilon\}^{-2} \leq H^{\langle j \rangle}_{\epsilon}$.
\end{enumerate}
\end{lemma}

\begin{proof}
We start with the first item.  The second inequality is straightforward by the definition of $H^{\langle j \rangle}$.  For the first inequality, for each arm $i \in V$, we consider two cases.
\begin{enumerate}
\item $\theta_i \ge \theta_{[k]}(V)$:  It holds that $\Delta^{\langle k \rangle}_i(V) = \theta_i - \theta_{[k + 1]}(V) \geq \theta_i - \theta_{[j + 1]} = \Delta^{\langle j\rangle}_i$.

\item $\theta_i \le \theta_{[k + 1]}(V)$: It holds that $\Delta^{\langle k\rangle}_i(V) = \theta_{[k]}(V) - \theta_i \geq \theta_{[j]} - \theta_i = \Delta^{\langle j\rangle}_i$.
\end{enumerate}
We thus have 
$$H^{\langle k\rangle}(V) = \sum\limits_{i \in V}\left(\Delta^{\langle k\rangle}_i(V)\right)^{-2} \leq \sum\limits_{i \in V} \left(\Delta^{\langle j\rangle}_i\right)^{-2}.$$

The second item follows from the same line of arguments.
\end{proof}

The next lemma gives a connection between instance complexities under different pivots.  

\begin{lemma}
\label{lem:truncate}
For any $t \in \{1, \ldots, n\}$,  $H^{\langle t \rangle}_{\Delta_{[t]}^{\langle m \rangle}} \leq 4 H^{\langle m\rangle}$.
\end{lemma}

\begin{proof}
We consider $t \le m$ and $t > m$ separately.  In the case that $t \le m$, for any $i \in I$ we consider three cases.
\begin{enumerate}
\item  $\theta_i < \theta_{[m]}$: In this case we have 
$$
\max\{\Delta^{\langle t\rangle}_i, \Delta_{[t]}^{\langle m \rangle}\} \geq \Delta^{\langle t\rangle}_i = \theta_{[t]} - \theta_i \geq \theta_{[m]} - \theta_i = \Delta^{\langle m\rangle}_i.
$$

\item  $\theta_{[m]} \le \theta_i < \theta_{[m+1]} + 2 \Delta_{[t]}^{\langle m \rangle}$: In this case we have 
$$
\max\{\Delta^{\langle t\rangle}_i, \Delta_{[t]}^{\langle m \rangle}\} \geq \Delta_{[t]}^{\langle m \rangle} \geq \frac{\theta_i - \theta_{[m + 1]}}{2} \geq \frac{\Delta^{\langle m\rangle}_i}{2}.
$$

\item  $\theta_{[m+1]} + 2\Delta_{[t]}^{\langle m \rangle} \le \theta_i$: In this case we have
\begin{eqnarray*}
\max\{\Delta^{\langle t\rangle}_i, \Delta_{[t]}^{\langle m \rangle}\} \geq  \Delta^{\langle t \rangle}_i &=& \theta_i - \theta_{[t + 1]} = \theta_i - \theta_{[m + 1]} + \theta_{[m + 1]} - \theta_{[t + 1]} \\
&\ge&  \frac{\theta_i - \theta_{[m + 1]}}{2} + \frac{\theta_i - \theta_{[m + 1]}}{2}  - (\theta_{[t + 1]} - \theta_{[m + 1]})  \\
&\geq& \frac{\Delta^{\langle m \rangle}_i}{2} + \Delta_{[t]}^{\langle m \rangle} - (\theta_{[t + 1]} - \theta_{[m + 1]}) \geq \frac{\Delta^{\langle m \rangle}_i}{2}.
\end{eqnarray*}
Therefore,
$$
H^{\langle t\rangle}_{\Delta_{[t]}^{\langle m \rangle}} = \sum\limits_{i \in I} \max\left\{\Delta^{\langle t\rangle}_i, \Delta_{[t]}^{\langle m \rangle}\right\}^{-2} \leq \sum\limits_{i \in I}4  \left(\Delta^{\langle m\rangle}_i\right)^{-2} \leq 4  H^{\langle m\rangle}.
$$
\end{enumerate}

In the case when $t > m$, the proof is symmetric by considering the following three cases: $(1')$ $\theta_i > \theta_{[m+1]}$, $(2')$ $\theta_{[m]} \ge \theta_i > \theta_{[m+1]} + 2 \Delta_{[t]}^{\langle m \rangle}$, and $(3')$ $\theta_{[m+1]} + 2\Delta_{[t]}^{\langle m \rangle} \ge \theta_i$.
\end{proof}

The following simple fact gives an upper bound of the contribution (to the instance complexity) of an arm that is not very close to the pivot.

\begin{lemma}
\label{lem:far}
For any $t \in \{1, \ldots, n\}$, if $t \le m - z$ or $t \ge m + z$, then 
$$
\left(\Delta_{[t]}^{\langle m \rangle}\right)^{-2} \le H^{\langle m\rangle} / z.
$$
\end{lemma}

\begin{proof}
In the case that $t \le m - z$, there are at least $z$ arms $i$ such that $\theta_{[t]} \ge \theta_i \ge \theta_{[m]}$, or $\Delta_i^{[m]} \le \Delta_{[t]}^{\langle m \rangle}$.  Consequently we have
$$
H^{\langle m\rangle} = \sum\limits_{i \in I}{\left(\Delta^{\langle m\rangle}_i\right)^{-2}} \geq z \cdot \left(\Delta_{[t]}^{\langle m \rangle}\right)^{-2}.
$$
The case $t \ge m+z$ can be proved in the same way.
\end{proof}

We need two centralized algorithms \cenT\ and \cenB\ for computing $(\eps, m)$-top/bottom arms.  We leave their detailed description to Section~\ref{sec:PAC}.  The following lemma summarizes the guarantees of these two algorithms; it is a direct consequence of Lemma~\ref{lem:cenT} and Lemma~\ref{lem:cenB}, which will be presented and proved in Section~\ref{sec:PAC}.

\begin{lemma}
\label{lem:centra-1}

Let $I$ be a set of $n$ arms, $m \in\{1, \dotsc, n - 1\}$, and $\eps \in (0,1)$ be an approximation parameter.  Let 
$$T_1(I, a, \eps, \delta) = c_{1} H^{\langle a\rangle}_{\epsilon / 2}(I) \cdot \log \left({H^{\langle a \rangle}_{\epsilon / 2}(I) / \delta}\right)
$$ for a universal constant $c_{1}$. 
We have that 
\begin{itemize}
\item If $T \ge T_1(I, m, \eps, \delta)$ then
\cenT$(I, m, T, \delta)$ with probability at least $1 - \delta$, returns $m$ arms each of which is $(\eps, m)$-top in $I$ using at most $T$ time steps.

\item If $T \ge T_1(I, n - m, \eps, \delta)$ then \cenB$(I, m, T, \delta)$ with probability at least $1 - \delta$, returns $m$ arms each of which is $(\eps, m)$-bottom in $I$ using at most $T$ time steps.
\end{itemize}
\end{lemma}

The following lemma says that there is a simple collaborative algorithm \colN\ for top-$m$ arm identification that uses $O(\log n)$ rounds of communication. Note that this bound is still much larger than our final target $O(\log\log m + \log K)$ rounds. \colN\ is a simple modification of a centralized algorithm in \cite{BWV13}, and will be described in details in Section~\ref{sec:simple-alg}. Lemma~\ref{lem:centra-2} is a direct consequence of Lemma~\ref{lem:colN}, which will be presented and proved in Section~\ref{sec:simple-alg}.

\begin{lemma}
\label{lem:centra-2}
Let $I$ be a set of $n$ arms, and $m \in \{1, \ldots, n - 1\}$.  Let 
\begin{equation}
\label{eq:T2}
T_2(I, m, \delta) = c_{2} \cdot \frac{H^{\langle m \rangle}(I)}{K} \cdot \log{n} \cdot \log\frac{n}{\delta}
\end{equation}
for a universal constant $c_{2}$.
There is a collaborative algorithm \colN$(I, m, T)$ such that if $T \ge T_2(I, m, \delta)$ then with probability at least $1 - \delta$, one computes the set of top-$m$ arms of $I$ using at most $T$ time steps and $O(\log n)$ rounds.
\end{lemma}

\subsection{Special Time Horizon $T$}

In this section we prove the following theorem concerning a special time horizon $T$.

\begin{theorem}
\label{thm:ub}

Let $I$ be a set of $n$ arms, and $m \in \{1, \dotsc, n-1\}$. Let 
\begin{equation}
\label{eq:T0}
T_0 = c_{0} \cdot \frac{H^{\langle m \rangle}}{K} \cdot \left(\log \left(H^{\langle m\rangle} \cdot K\right)+\log^2 n \right)\cdot \log\log n 
\end{equation}
for a large enough constant $c_{0}$.
There exists a collaborative algorithm \colM$(I, m, T)$\ that computes the set of top-$m$ arms of $I$ with probability at least $0.99$ when $T \ge T_0$, and uses at most $T$ time steps and $O(\log\frac{\log n}{\log K} + \log K)$ rounds of communication.
\end{theorem}


\begin{algorithm}[t]
\caption{\colM$(I, m, T)$}
\label{alg:main}
\KwIn{a set of $n$ arms $I$, parameter $m$, and time horizon $T$.}
\KwOut{the set of top-$m$ arms of $I$.}
Let $R$ be the global upper bound on the number of rounds \\ and $\delta$ be also the global parameter equal to $1/(100 R)$\;
$q \gets 4  K \sqrt{n\log \left(n R \right)}$\;

\eIf{$n > K^{10}$}{
$\Acc \gets \emptyset, \Rej \gets \emptyset$\;
randomly assign each arm in $I$ to one of the $K$ agents, and let $I_{i}$ be the set of arms assigned to $i$-th agent  \footnotemark \label{ln:a-0}\;
\If{$m > q$}{
$\ell \gets (m - q) / K$\;
\For{agent $i = 1$ to $K$}{
    $\Acc_i \gets \text{\cenT}\left(I_i, \ell, \frac{T}{4R}, \frac{\delta}{2K}\right)$  \label{ln:a-1}\;
}
$\Acc \gets \bigcup_{i=1}^{K} \Acc_i$\;
}
\If{$n - m > q$}{
$r \gets (n - m - q)/K$\;
\For{agent $i = 1$ to $K$}{
    $\Rej_i \gets \text{\cenB}\left(I_i, r, \frac{T}{4R}, \frac{\delta}{2K}\right)$\; \label{ln:a-2}
}
$\Rej \gets \bigcup_{i=1}^K \Rej_i$\;
}
\KwRet{$\Acc \bigcup \text{\colM}(I \setminus (\Acc \cup \Rej), m - |\Acc|, T)$} \label{ln:a-3}\;
}{\KwRet{\colN$(I, m, T/2)$ \label{ln:a-4}.}}
\end{algorithm}


\paragraph{Algorithm and Intuition.}
Our algorithm is described in  Algorithm~\ref{alg:main}.  
Note that we have used {\em recursion} instead of {\em iteration} to omit a superscript $r$. But we still call each recursive step a {\em round}.  

Let us briefly describe Algorithm~\ref{alg:main} in words. At the beginning of each round we first randomly partition the set of arms to the $K$ agents. Then each agent tries to identify a subset of arms $\Acc_i$ of size $\ell \approx (m/K - \sqrt{n})$ to be included to $\Top_m$, and a subset of arms $\Rej_i$ of size $r \approx ((n-m)/K - \sqrt{n})$ to be pruned.  The intuition to introduce the additive $\sqrt{n}$ term is that by a concentration bound, we have with a good probability that at least $\ell$ true top-$m$ arms will be assigned to each agent, and similarly at least $r$ non-top-$m$ arms will be assigned to each agent.  However, even with this fact, we still cannot guarantee that each agent can identify the top and bottom arms successfully given its limited budget, which is approximately $H^{\langle m \rangle}/K$. Such a budget in some sense demands that the global instance complexity is {\em evenly} divided into the $K$ agents, which is not necessary true.  We thus adopt a PAC algorithm for top-$m$ arm identification which returns a set of $\ell$ $(\eps, \ell)$-top arms at each agent $A_i$, where $\eps$ is a random variable which, with a high probability, is smaller than the gap between the $\ell$-th top arm {\em locally} at $A_i$ and that of the $m$-th {\em global} top arm. In this way we can guarantee that it is safe to include each $\Acc_i$ that $A_i$ computes into $\Top_m$.  By essentially the same arguments, we can show that it is safe to prune the set of bottom arms $\Rej_i$.

The following lemma is critical for the correctness of the algorithm.

\begin{lemma}
\label{lem:correctness}

if $T \ge T_0$ (defined in~\refeq{eq:T0}). When running Algorithm~\ref{alg:main} \colM$(I, m, T)$ , we have that at each round,
\begin{equation}
\label{eq:contain}
\Pr\left[(\Acc \subseteq \Top_m) \wedge (\Rej \subseteq I \setminus \Top_m) \right]\geq 1 - {1}/{(200 R)},
\end{equation}
where $\Acc$ and $\Rej$ are defined in Algorithm~\ref{alg:main}.
\end{lemma}

Before proving Lemma~\ref{lem:correctness}, we first show that Lemma~\ref{lem:correctness} implies Theorem~\ref{thm:ub}.

\begin{proof}[Proof of Theorem~\ref{thm:ub}]  
W.l.o.g.\ we assume that $K = \omega(\log\log n)$. Note that when $K = O(\log\log n)$, we have $T_0 \ge K \cdot T_2(I, m, 1/100)$, and thus each agent can simply solve the problem independently using the centralized algorithm in \cite{BWV13}.

Let $m' = m - \abs{\Acc}$. If $\Acc \subseteq \Top_m$ and $\Rej \subseteq I \setminus \Top_m$, then we have
\begin{enumerate}
\item $\Top_m = \Acc \cup \Top_{m'}(I \setminus (\Acc \cup \Rej))$.

\item $H^{\langle m' \rangle}(I \setminus (\Acc \cup \Rej)) \le H^{\langle m \rangle}$.
\end{enumerate}
The first item is obvious.  The second item is due to Lemma~\ref{lem:sandwich}.  The second item ensures that the recursion goes through under the same time horizon.

By the first item, Lemma~\ref{lem:centra-2} (note that $T/2 \ge T_2(I, m, 1/200)$), and a union bound, we have that with error at most $(1/(200 R) \cdot R + 1/200) = 1/100$, Algorithm~\ref{alg:main} computes $\Top_{m}$.

Now we analyze the running time.  Under the condition that at each round we have $\Acc \subseteq \Top_m$ and $\Rej \subseteq I \setminus \Top_m$, it follows that when $n > K^8$ and $K = \omega(\log\log n)$,
$$
n - \abs{\Acc} - \abs{\Rej} \le 2 \cdot 4  K \sqrt{n \log\left(n R \right)} \le n^{7/8}.
$$
Therefore after $R = 10\log\left(\frac{\log n}{10 \log K}\right)$ rounds, we have $n^{(7/8)^{R}} \le K^{10}$. Consequently Algorithm~\ref{alg:main} must have already reached Line~\ref{ln:a-4}. The algorithm \colN\ in Line~\ref{ln:a-4} takes $O(\log K^{10}) = O(\log K)$ rounds by Lemma~\ref{lem:colN} (setting $R = \log n$ where $n = K^{10}$ here).   Thus the total number of rounds is bounded by $O\left(\log\left(\frac{\log n}{\log K}\right) + \log K\right)$.  
\end{proof}

In the rest of this section we will prove Lemma~\ref{lem:correctness}.  
\smallskip

Let $q = 4 K \sqrt{n\log \left(n R \right)}$, $\ell = \left(m - q\right)/K$ and $r = \left(\left(n - m\right) - q\right)/K$ be defined in Algorithm~\ref{alg:main}. Further, define $a = m - q/2$ and $b = \left(n - m\right) - q/2$.  

\footnotetext{This randomness can be precomputed and stored at each agent.}

The following lemma concerns properties of the random partition in Line~\ref{ln:a-0} in Algorithm~\ref{alg:main}.

\begin{lemma}
\label{lem:partition}
Let $V$ be a random subset of $I\ (n > K^{10})$ by taking each arm independently with probability $1/K$.  We have
\begin{enumerate}
\item If $m>q$, then $\Pr\left[\left(|V| \geq \ell \right) \wedge \left(\theta_{[\ell]}(V) \geq \theta_{[a]}\right)\right] \geq 1 - \frac{1}{1600 K R}$.  

\item If $n - m > q$, then $\Pr\left[\left(|V| \geq r \right) \wedge \left(\theta_{[|V| - r + 1]}(V) \leq \theta_{[n - b + 1]}\right)\right] \geq 1 - \frac{1}{1600 K R}$.
\end{enumerate}
\end{lemma}

\begin{proof}
We focus on the first item; the second item is symmetric and can be proved by similar arguments.

For each $i \in \{1, \ldots, a\}$, we define a random variable $X_i$ which is $1$ if the arm with mean $\theta_{[i]}$ lies in the set $V$, and $0$ otherwise.  Let $X = \sum_{i =1}^a X_i$.  Thus $\bE[X] = a/K$. By Chernoff-Hoeffding (Lemma~\ref{lem:chernoff}) we have
$$
\Pr[X < \ell] = \Pr \left[X < \frac{a}{K} - \frac{q}{2 K}\right] \leq \exp\left(\frac{-8 {n \log \left(n R\right)}}{a}\right) \leq \frac{1}{1600 K R}.
$$
Thus with probability at least $\left(1 - \frac{1}{1600 K R}\right)$, $V$ contains at least $\ell$ arms with mean at least $\theta_{[a]}$.
\end{proof}

The next lemma connects the global instance complexity with the local instance complexity. 

\begin{lemma}
\label{lem:local-global}
Let $V$ be a random subset of $I\ (n > K^{10})$ by taking each arm independently with probability $1/K$.  Let $P$ be a random variable such that $\theta_{[P]} = \theta_{[\ell]}(V)$, and $Q$ be a random variable such that $\theta_{[Q]} = \theta_{[\abs{V}-r+1]}(V)$.
We have
\begin{enumerate}
\item If $m>q$, then $\Pr\left[(\abs{V} \ge \ell)  \wedge  \left(H_{\Delta_{[P]}^{\langle m \rangle}}^{\langle \ell \rangle}(V) \le 5H^{\langle m \rangle}/K \right) \right] \ge 1 - \frac{1}{800 K R}$.  
\item If $n - m > q$, then $\Pr\left[(\abs{V} \ge r) \wedge \left(H_{\Delta_{[Q]}^{\langle m \rangle}}^{\langle r \rangle}(V) \le 5H^{\langle m \rangle}/K \right)\right] \ge 1 - \frac{1}{800 K R}$.  
\end{enumerate}
\end{lemma}

\begin{proof}
We only need to prove the first item. The second item follows by symmetry. 

When $m > q$, let $\chi$ denote the event $(|V| < \ell) \vee (\theta_{[\ell]}(V) < \theta_{[a]})$. By Lemma~\ref{lem:partition} we have that $\chi$ happens with probability at most $1/(1600 K R)$. 
We have 
\begin{eqnarray}
\Pr\left[\left(|V| < \ell\right) \lor \left(H_{\Delta_{[P]}^{\langle m \rangle}}^{\langle \ell \rangle}(V) > \frac{5H^{\langle m \rangle}}{K} \right)\right] &\le& \Pr\left[\bar{\chi} \land \left(H_{\Delta_{[P]}^{\langle m \rangle}}^{\langle \ell \rangle}(V) > \frac{5H^{\langle m \rangle}}{K} \right)\right] + \Pr[\chi] \nonumber \\
&\le& \Pr\left[\bar{\chi} \land \left(H_{\Delta_{[P]}^{\langle m \rangle}}^{\langle \ell \rangle}(V) > \frac{5H^{\langle m \rangle}}{K} \right)\right] + \frac{1}{1600KR}.
\label{eq:chi}
\end{eqnarray}
Apply the law of total probability to the first term:
\begin{eqnarray}\label{eq:total}
\Pr\left[\bar{\chi} \land \left(H_{\Delta_{[P]}^{\langle m \rangle}}^{\langle \ell \rangle}(V) > \frac{5H^{\langle m \rangle}}{K} \right)\right]  &=&  \sum_{p=1}^a \Pr\left[\left(H^{\langle \ell \rangle}_{\Delta_{[p]}^{\langle m \rangle}}(V) > \frac{5H^{\langle m \rangle}}{K} \right)\land (P = p) \right] \\
 & \le & \sum_{p=1}^a \Pr\left[\left(\sum_{i \in V}\max\left\{\Delta_i^{\langle p \rangle}, \Delta_{[p]}^{\langle m \rangle}\right\}^{-2} > \frac{5H^{\langle m \rangle}}{K} \right) \land (P = p) \right] \nonumber \\
& = & \sum_{p=1}^a \Pr\left[\sum_{i \in V}\max\left\{\Delta_i^{\langle p \rangle} , \Delta_{[p]}^{\langle m \rangle}\right\}^{-2} > \frac{5H^{\langle m \rangle}}{K} \right], \label{eq:c-3}
\end{eqnarray}
where the inequality is due to Lemma~\ref{lem:sandwich}.

For each $i \in I$, define $X_i = \max\left\{\Delta_i^{\langle p \rangle}, \Delta_{[p]}^{\langle m \rangle}\right\}^{-2}$ if $i \in V$, and $X_i = 0$ otherwise.  Let $X = \sum_{i \in I} X_i$.  We thus have $\bE[X] = {H^{\langle p \rangle}_{\Delta_{[p]}^{\langle m \rangle}}/ K}$. By Lemma~\ref{lem:truncate} we have 
\begin{equation}
\label{eq:d-1}
\bE[X] \le \frac{4H^{\langle m \rangle}}{ K}.
\end{equation}  
By Lemma~\ref{lem:far} and $p \le a = m - q/2$, we have 
\begin{equation}
\label{eq:d-2}
X_i = \left[0, H^{\langle m\rangle}\left/\frac{q}{2} \right.\right].
\end{equation}
By (\ref{eq:d-1}), (\ref{eq:d-2}), and Chernoff-Hoeffding we have
\begin{eqnarray}
\label{eq:d-3}
\Pr\left[X > 5H^{\langle m \rangle}/K \right]  &=& \Pr\left[X > \bE[X] + H^{\langle m \rangle}/K \right] \nonumber \\
&\le& \exp\left(\frac{-2 \left(H^{\langle m \rangle}/K\right)^2}{n \left(2H^{\langle m\rangle}/q\right)^2 } \right) = \exp\left(\frac{-q^2}{2 K^2 n}\right) \nonumber \\
&\le& \frac{1}{1600n K R}. \nonumber
\end{eqnarray}
We thus have $(\ref{eq:c-3}) \le a \cdot {1}/{(1600n K R)} \le {1}/{(1600 K R)}$, which, together with \eqref{eq:chi}, gives the first item of the lemma.
\end{proof}

Now we are ready to prove Lemma~\ref{lem:correctness}.
\begin{proof}[Proof of Lemma~\ref{lem:correctness}]
We first analyze the probability that $\Acc \subseteq \Top_m$.

Again let $P$ be the random variable such that $\theta_{[P]} = \theta_{[\ell]}(I_i)$ for a partition $I_i$.  Since at Line~\ref{ln:a-1} of Algorithm~\ref{alg:main} we call \cenT\ with time budget ${T/4R}$, we have that if $H_{\Delta_{[P]}^{\langle m \rangle}}^{\langle \ell \rangle}(V) \le 5H^{\langle m \rangle}/K$ then 
$$\frac{T}{4R} \ge T_1\left(I_i, \ell, \Delta_{[P]}^{\langle m \rangle} / 4, 1/(800KR)\right),$$
and with probability at least $1/(800KR)\cdot K = 1/(800R)$, it holds that \cenT$\left(I_i, \ell, \frac{T}{4R}\right)$ succeeds for all $i \in \{1, \ldots, K\}$, which, together with the first item of Lemma~\ref{lem:local-global}, gives that \[\Pr[\Acc \subseteq \Top_m] \ge 1 - \frac{1}{400 R}\,.\]

By symmetric arguments we can also show that $\Pr[\Rej \subseteq I \setminus \Top_m] \ge 1 - 1/(400 R)$.
\end{proof}

\subsection{General Time Horizon $T$}
\label{sec:general-T}

Theorem~\ref{thm:ub} only achieves a constant error probability for a special case of the time horizon $T = \tilde{\Theta}(T_0)$ where $T_0 = H^{\langle m \rangle}/K$.  Our next goal is to consider general time horizon $T \ge T_0$, and try to make the error probability decrease exponentially with respect to $T/T_0$.  More precisely, we have the following theorem.

\begin{theorem}
\label{thm:ub-general}
Let $I$ be a set of $n$ arms, and $m \in \{1, \ldots, n - 1\}$. Let $T$ be a time horizon.
There exists a collaborative algorithm \colMG\ that computes the set of top-$m$ arms of $I$ with probability at least 
$$1 - n \cdot \exp\left(-\Omega\left(\frac{TK}{H^{\langle m\rangle}\cdot (\log({H^{\langle m\rangle} K}) + \log^2{n})\cdot \log\log{n} \cdot \log^2\left(TK/H^{\langle m\rangle}\right)}\right)\right)$$
using at most $T$ time steps and $O\left(\log\frac{\log n}{\log K} + \log K\right)$ rounds.
\end{theorem}

\paragraph{High Level Idea.}
A standard technique to achieve an error probability that is exponentially small in terms of $T/T_0$ is to perform parallel repetition and then take the majority.  This is straightforward if we know the value $T_0$. Unfortunately, $T_0$ depends on the instance complexity which we do not know in advance.  A standard trick to handle this issue is to use the doubling method. That is, we guess $T_0 = 1, 2, 4, \ldots$, and for each value we repeat $T/T_0$ times (ignoring logarithmic factors). We know that one of these values is very close to the actual $T_0$. We hope that this value is the {\em first} value in $\{1, 2, 4, \ldots \}$ for which the $T/T_0$ runs of \colM$(I, m, T)$ contain a majority output.

The main issue in this approach is that when $T \le T_0$, the output of the algorithm can be {\em consistently} wrong, which leads to a wrong majority.  Note that we do not have much control on the output of the algorithm when the time horizon is very small.

We handle this issue by introducing a concept called {\em top-$m$ certificate}. We require each algorithm for top-$m$ arm identification to output a pair $(S, \{\tilde{\theta}_i\}_{i \in I})$,  where $S$ is a subset of $I$ of size $m$ and $\{\tilde{\theta}_i\}_{i \in I}$ are the estimated means for {\em all} arms in $I$ (not just those in $S$).  We say a pair  $(S, \{\tilde{\theta}_i\}_{i \in I})$ is a top-$m$ certificate if it can pass an additional verification step which checks whether $S$ is indeed the set of top-$m$ arms of $I$ given the estimated means $\{\tilde{\theta}_i\}_{i \in I}$.  With such a verification step at hand, we do not need to worry about the case that \colM\ will output a wrong answer when $T$ is too small, since a wrong output will simply {\em not} pass the verification step.  Finally, we make sure that this verification step is perfectly parallelizable and thus fit in our time budget. 

In the rest of this section we will first give the definition of the top-$m$ certificate and describe the verification algorithm, and then give the collaborative algorithm for general time horizon $T$.

\subsubsection{Top-$m$ Certificate}
\label{sec:certificate}

\begin{definition}[Top-$m$ Certificate]  Let $S \subseteq I$.
We say $(S,  \{\tilde{\theta}_i\}_{i \in I})$ is a top-$m$ certificate of $I$ if
\begin{equation*}
\left(S = \Top_m\right) \wedge \left(\forall i \in I : \abs{\tilde{\theta}_i - \theta_i} < \Delta^{\langle m\rangle}_i / 4 \right).
\end{equation*}
\end{definition}

\begin{observation}
If $(S, \{\tilde{\theta}_i\}_{i \in I})$ is a top-$m$ certificate, then for any $i \in S$ and $j \in I \backslash S$, we have $\tilde{\theta}_i > \tilde{\theta}_j$. 
\end{observation}

Given an arbitrary pair $(S, \{\tilde{\theta}_i\}_{i \in I})$, we can design an algorithm to verify whether  $(S, \{\tilde{\theta}_i\}_{i \in I})$ is a top-$m$ certificate of $I$; see Algorithm~\ref{alg:verify} \verM.  We note that Algorithm~\ref{alg:verify} can be easily implemented in $O(1)$ rounds since the number of pulls on each arm is determined in advance and can thus be fully parallelised.

\begin{algorithm}[t]
\caption{\verM$(I, m, S, \{\tilde{\theta}_i\}_{i \in I}, \gamma, T)$}
\label{alg:verify}
\KwIn{a set of $n$ arms $I$, parameter $m$, pair $(S, \{\tilde{\theta}_i\}_{i \in I})$, parameter $\gamma$, and time horizon $T$.}
\KwOut{the set of top-$m$ arms, or $\bot$.}
Set $\ell \gets \arg\min_{i \in S} \tilde{\theta}_i$ and $r \gets \arg\max_{i \in I \backslash S} \tilde{\theta}_i$\; 
for each $i \in S$ set $\Delta_i \gets \tilde{\theta}_i - \tilde{\theta}_r$, and for $i \in I \backslash S$ set $\Delta_i \gets \tilde{\theta}_\ell - \tilde{\theta}_i$\; 
\lIf{$\exists i \text{ s.t. } \Delta_i \le 0$ or $|S| \neq m$}{\KwRet{$\bot$}} 
\If{$\sum\limits_{i \in I}{\left(64 \gamma \Delta_i^{-2}\right)} \le K T$
\label{ln:threshold}}{
pull $i$-th arm for $64 \gamma \Delta^{-2}_i$ times and let $\hat{\theta}_i$ be the empirical mean\; 
\lIf{$\min\limits_{i \in S}\{\hat{\theta}_i - \Delta_i/4\} > \max\limits_{i \in I \backslash S}\{\hat{\theta}_i + \Delta_i/4\}$ \label{ln:separation}}
{\KwRet{$S$}}
}
\KwRet{$\bot$.}
\end{algorithm}

The following lemma shows that if $(S, \{\tilde{\theta}_i\}_{i \in I})$ is indeed a certificate and $T$ is large enough, then \verM\ returns the set $S$ with a good probability.

\begin{lemma}
\label{lem:verify}
For any $I, m, S, \{\tilde{\theta}_i\}_{i \in I}, \gamma, T$, we have
\begin{equation}\label{eq:f-0}	
	\Pr[\text{\verM}(I, m, S, \{\tilde{\theta}_i\}_{i \in I}, \gamma, T) \not\in \{\Top_{m}, \bot\}] \le 2 n \cdot e^{-\gamma}\,.
\end{equation}
Moreover, if $(S, \{\tilde{\theta}_i\}_{i \in I})$ is a top-$m$ certificate of $I$, and $T \ge 200 \gamma  H^{\langle m \rangle} / K$ then 
\begin{equation}\label{eq:f-1}	
    \Pr[\text{\verM}(I, m, S, \{\tilde{\theta}_i\}_{i \in I}, \gamma, T) \neq \Top_m] \le 2 n \cdot e^{-\gamma}.
\end{equation}
\end{lemma}

\begin{proof}
By Chernoff-Hoeffding, for any $i \in I$, after $64\gamma\Delta_i^{-2}$ pulls, we have $\Pr[|\hat{\theta}_i - \theta_i| > {\Delta_i/8}] \leq 2 e^{-\gamma}$. By a union bound we have
\begin{equation}
\label{eq:f-1-5}
\Pr[\exists i : |\hat{\theta}_i - \theta_i| > \Delta_i/8] \leq 2 n \cdot e^{-\gamma}.
\end{equation}
We now show that if 
\begin{equation}
\label{eq:f-2}
\forall i : |\hat{\theta}_i - \theta_i| \le \Delta_i/8,
\end{equation}
then Algorithm~\ref{alg:verify} can only output $S$ when $S = \Top_m$.  

We prove by contradiction. Suppose  Algorithm~\ref{alg:verify} outputs $S$ when $S \neq \Top_m$, then there must exist a pair $(i, j)$ such that $i \in S \backslash \Top_m$ and $j \in \Top_m \backslash S$, and consequently $\theta_i < \theta_j$. Meanwhile, by Line~\ref{ln:separation} of Algorithm~\ref{alg:verify} we have
\begin{equation}
\label{eq:f-3}
\hat{\theta}_i - \Delta_i/4 > \hat{\theta}_j + \Delta_j/4.
\end{equation}
Combining (\ref{eq:f-2}) and (\ref{eq:f-3}) we have
$$\theta_i > \theta_i - \Delta_i/8 > \theta_j + \Delta_j/8 > \theta_j,$$ 
which contradicts to the choices of $(i, j)$.  This proves (\ref{eq:f-0}).

We next prove (\ref{eq:f-1}).  If $(S, \{\tilde{\theta}_i\}_{i \in I})$ is a top-$m$ certificate, then by (\ref{eq:f-2}) we have 
\begin{equation*}
{\Delta^{\langle m\rangle}_i}/{2}\le \Delta_{i} \le {3 \Delta^{\langle m \rangle}_i}/{2}.
\end{equation*}
Thus if $T \ge 200 \gamma H^{\langle m \rangle} / K$, then 
\begin{equation*}
\sum\limits_{i \in I}{\left(64 \gamma \Delta_i^{-2}\right)} \le 200\gamma \sum_{i \in I} \left(\Delta_i^{\langle m \rangle}\right)^{-2} \le K T.
\end{equation*}
We thus only need to show that 
\begin{equation}
\min\limits_{i \in \Top_m}\{\hat{\theta}_i - \Delta_i/4\} > \max\limits_{i \not\in \Top_m}\{\hat{\theta}_i + \Delta_i/4\}\,.
\end{equation}

We again prove by contradiction.  Suppose that there exists a pair $(i, j)$ such that $i \in \Top_m$, $j \not\in \Top_m$, and $\hat{\theta}_i - \Delta_i/4 \le \hat{\theta}_j + \Delta_j/4$, then by (\ref{eq:f-2}) we have 
\begin{equation}
\theta_i - \theta_j \le \frac{3}{8}(\Delta_i + \Delta_j) \le \frac{3}{8}\left(\Delta^{\langle m \rangle}_i + \Delta^{\langle m\rangle}_j\right),
\end{equation}
which contradicts to the fact that  $\frac{1}{2}\left(\Delta^{\langle m \rangle}_i + \Delta^{\langle m\rangle}_j\right) \le \theta_i - \theta_j$.
\end{proof}

For technical reasons we need the following lemma, which says that \verM\ is very likely to output $\perp$ when the time horizon $T$ is small.

\begin{lemma}
\label{lem:verify-bad}
For any $I$, $m$, $S$, $\{\tilde{\theta}_i\}_{i \in I}$, and $\gamma$,
if $T < \frac{\gamma}{16} H^{\langle m \rangle} / K$, then \[\Pr[\text{\verM}(I, m, S, \{\tilde{\theta}_i\}_{i \in I}, \gamma, T) = \bot] \geq 1 - 2n \cdot \exp(-\gamma)\,.\]
\end{lemma}

\begin{proof}
We can assume that $\Top_m = S$, and for any $i \in S$ and $j \not\in S$ we have $\tilde{\theta}_i > \tilde{\theta}_j$, since otherwise \verM\ will simply output $\perp$.

If for all $i$ we have $\Delta_i \le 32 \Delta^{\langle m \rangle}_i$, then $\sum\limits_{i \in I}{\Delta^{-2}_i} \ge H^{\langle m\rangle} / 1024$, and consequently 
$$
\sum_{i \in I}\left(64\gamma \Delta_i^{-2}\right) \ge \frac{\gamma}{16} H^{\langle m \rangle} > KT.
$$
In this case, according to Line~\ref{ln:threshold} of Algorithm~\ref{alg:verify}, \verM\ will output $\perp$.

Now consider the case that there exists $i$ such that $\Delta_i > 32 \Delta^{\langle m\rangle}_i$. We first consider the case $i \in S$. By (\ref{eq:f-1-5}) we have that with probability $1 - 2n \cdot \exp(-\gamma)$, the event holds: $\forall{i \in I} : |\hat{\theta}_i - \theta_i| < \Delta_i / 8$; denote this event by $\E_1$.
Consequently,
\begin{equation}
\label{eq:e-1}
\hat{\theta}_i - \Delta_i/4 \le \theta_i - \Delta_i/8 \le \theta_i - 4\Delta_i^{\langle m \rangle}.
\end{equation} 
Consider $j$ such that $\theta_j = \theta_{[m + 1]}$.  When $\E_1$ holds, we have
\begin{equation}
\label{eq:e-2}
\hat{\theta}_j + \Delta_j / 4 \ge \theta_j.
\end{equation}
Note that at Line~\ref{ln:threshold}, Algorithm~\ref{alg:verify} returns a set only if $\hat{\theta}_i - \Delta_i/4 > \hat{\theta}_j + \Delta_{j}/4$, which, by (\ref{eq:e-1}) and (\ref{eq:e-2}), is equivalent to $\theta_i - 4\Delta_i^{\langle m \rangle} > \theta_j$.  We now have
\begin{equation*}
\label{eq:e-3}
\Delta_i^{\langle m \rangle} = \theta_i - \theta_j > 4\Delta_i^{\langle m \rangle}.
\end{equation*}
A contradiction.

The case that $i \in I \backslash S$ can be proved by essentially the same arguments.
\end{proof}

\subsubsection{Algorithm for General Time Horizon}
\label{sec:general}

In this section we present an algorithm for general time horizon.  We first slightly augment Algorithm~\ref{alg:main} \colM\ so that it also outputs an estimate of the mean (i.e., the empirical mean) of each of the $n$ arms. Now the output of \colM\ is a top-$m$ certificate $(S, \{\tilde{\theta}_i\}_{i \in I})$.

We have the following lemma regarding \colM.  The proof is straightforward based on the properties of \cenT\ and \cenB, which can be found in Section~\ref{sec:ub-aux}.

\begin{lemma}
\label{lem:basic-certificate}
If $T \ge T_0$, then \colM$(I, m, T)$ is able to output a top-$m$ certificate of $I$ with probability at least $0.99$. 
\end{lemma}

Our final algorithm for the general time horizon is described in Algorithm~\ref{alg:general}.   It follows the guess-and-verify framework mentioned earlier.  Now we are ready to prove the main theorem of this section.

\begin{algorithm}[t]
\caption{\colMG$(I, m, T)$}
\label{alg:general}
\KwIn{a set of $n$ arms $I$, parameter $m$, and time horizon $T$.}
\KwOut{the set of top-$m$ arms of $I$.}
\For{$s = 1, 2, \dotsc$}{
run $4^s$ copies of \colM$\left(I, m, 3T/(\pi^2 s^2 4^{s})\right)$, and record the returned pair $(S, \{\tilde{\theta}_i\}_{i \in I})$\; 
let $S^{(s)}$ be the most frequent answer of the $4^s$ returned sets of top-$m$ arms, and for each $i \in [n]$, let $\tilde{\theta}_i^{(s)}$ be the median of the $4^s$ estimated means for the $i$-th arm\;
run \verM$(I, m, S^{(s)}, \tilde{\theta}^{(s)}, 4^s, 3T / (\pi^2 s^2))$ and let $A_s$ be the output \label{ln:verify}\;
} 
let $s^{\star}$ be the maximum $s$ such that $A_s \neq \bot$\;
\If{we cannot find such an $s^{\star}$ after $T$ time steps}{\KwRet{an arbitrary set of $m$ arms}\;}
\Else{\KwRet{$A_{s^{\star}}$}.}
\end{algorithm}

\begin{proof}[Proof of Theorem~\ref{thm:ub-general}]
Let $s'$ be the largest $s \ge 1$ such that $3T/(\pi^2 s^2 4^s) \ge T_0$, and consequently $4^{s'} = \Omega\left(\frac{T/T_0}{\log^2(T/T_0)}\right)$. If no such $s$ exists then Theorem~\ref{thm:ub-general} holds trivially.   By Lemma~\ref{lem:basic-certificate} and the standard median trick we have that with probability at least $1- \exp(-4^{s'})$, $(S^{(s')}, \tilde{\theta}^{(s')})$ is a top-$m$ certificate.  We also have 
\begin{equation*}
3T/ (\pi^2 (s')^2) \geq 200 \cdot 4^{s'} \cdot \frac{H^{\langle m \rangle}}{K},
\end{equation*}
which, combined with Lemma~\ref{lem:verify}, guarantees that the call of \verM\ at Line~\ref{ln:verify} in Algorithm~\ref{alg:general} returns $\Top_m$ with probability at least $1 - 2n \cdot \exp(-4^{s'})$.  By a union bound, we have that 
\begin{equation}
\label{eq:g-1}
\Pr[A_{s'} = \Top_m] \ge 1 - 4n \cdot \exp(-4^{s'}).
\end{equation}

On the other hand, for each $s = s' + j\ (j \ge 1)$, by Lemma~\ref{lem:verify} we have
\begin{equation}
\label{eq:g-2}
\Pr[A_{s'} \not\in \{\Top_m, \perp\}] \le 2n \cdot \exp(-4^{s'+j}).
\end{equation}
Combining (\ref{eq:g-1}) and (\ref{eq:g-2}), we have
\begin{eqnarray*}
    \Pr[A_{s^*} = \Top_m] &\geq& \Pr[A_{s'} = \Top_m] - \sum\limits_{j = s' + 1}^{\infty}\Pr[A_{j} \not\in \{\Top_m, \bot\}] \\
    &\ge& 1 - 4n\sum_{s \ge s'} \exp(-4^{s}) \ge 1 - 8n \cdot \exp(-4^{s'}).
\end{eqnarray*}
Plugging the fact that $4^{s'} = \Omega\left(\frac{T/T_0}{\log^2(T/T_0)}\right)$, we have $\Pr[A_{s^*} = \Top_m] \ge 1 - n \cdot \exp\left(-\Omega\left(\frac{T/T_0}{\log^2(T/T_0)}\right)\right)$.

Finally, it is easy to see that \colMG\ can be implemented in $O(\log\frac{\log n}{\log K} + \log K)$ rounds of communication by Theorem~\ref{thm:ub}, because all runs of \colM\ can be done in parallel and \verM\ requires only the constant number of rounds of communication. 
\end{proof}

\subsection{An Improved Algorithm}
\label{sec:ub-improve}

In this section we further improve the round complexity of Algorithm~\ref{alg:main} to $O(\log\frac{\log m}{\log K} + \log K)$.

\paragraph{High Level Idea.}
The general idea to improve the $\log\log n$ term to $\log\log m$ in the round complexity is to first reduce the number of arms from $n$ to $\tilde{O}(m)$. This idea is relatively easy to implement if we only target for a constant error probability: We sample each of the $n$ arms with probability $1/m$, getting a subset $V$.  By an easy calculation, the event that $V$ contains {\em exactly} one top-$m$ arms of $I$ happens with a constant probability.  Conditioned on this event, the expected sub-instance complexity $H^{\langle 1 \rangle}(V)$ is upper bounded by $O(H^{\langle m \rangle}(I)/m)$.  Thus if we sample $\tilde{O}(m)$ sub-instances, and compute the best arm in each sub-instance , then the collection of the $\tilde{O}(m)$ best arms will be a {\em superset} of $\Top_m$ with a good probability.  

As before, our ultimate goal is to make the error probability exponentially small in terms of $T$.  To this end we also need to amplify the success probability of the aforementioned reduction.  Unfortunately, using a verification step as that in Section~\ref{sec:general-T} is not sufficient here, and once again we need new ideas. We note that this section is the most technically challenging part in the whole algorithm design of this paper.

To start with, we again try to guess the sub-instance complexity using a geometric sequence, and use parallel repetition to amplify the success probability. Similar as Section~\ref{sec:general}, the key is to avoid outputting a wrong majority when the guess is too small.  This is possible if the following statement holds: 
\begin{quote}
{\em For a randomly sampled sub-instance on which we compute the best arm, the probability of outputting any arm inside global $\Top_m$ is larger than that of any arm outside $\Top_m$}.
\end{quote} 
Unfortunately, this statement does not always hold.  We thus try to prove a slightly weaker version of the statement.  We show that the probability of outputting any arm inside global $\Top_m$ is {\em not much smaller} than that of any arm outside $\Top_m$.

To show this, we consider three cases.  Let $V$ be the randomly sampled sub-instance.  Let $T_0$ be our guessed complexity of $V$, and $H = H^{\langle 1 \rangle}(V)$ be the real expected complexity of $V$.  
\begin{enumerate}
\item
If $T_0 \ge \lambda H$ where $\lambda = \log^{\Theta(1)}(T_0 Kn)$, then if we apply the best arm identification algorithm from \cite{TZZ19} on $V$ with time budget $T$, with a good probability we can successfully identify the best arm.  Note that since our subsampling is uniform, for any pair of two arms $a, b$ where $a \in \Top_m, b \not\in \Top_m$, they will be included in the sub-instance with equal probability.  Therefore the probability of outputting $a$ must be at least that of outputting $b$.

\item
If $T_0 \le H$, then by applying \cite{TZZ19} on $V$ with time budget $T_0$, together with a verification step (Algorithm~\ref{alg:verify}), we can detect with a good probability that our guess $T_0$ is too small compared with $H$.

\item 
To handle the gap $H < T_0 < \lambda H$, we apply the following trick: We replace $T_0$ by $T_0$ or $T_0/\lambda$ with equal probability, so as to ``reduce'' this case to either the first or the second case, for which we know how to handle. In this way we can show that for $a \in \Top_m, b \not\in \Top_m$, the probability of outputting $a$ is at least {\em half} that of outputting $b$. We note that this intuitive description is not entirely precise, but it conveys the idea.
\end{enumerate}

In this section we prove the following theorem.

\begin{theorem}
\label{thm:ub-improve}
Let $I$ be a set of $n$ arms,  $m \in \{1, \ldots, n - 1\}$, and $T$ be the time horizon. There exists a collaborative algorithm that computes the set of top-$m$ arms of $I$ with probability at least $$1 - n \cdot \exp\left(-\Omega\left(\frac{KT}{H^{\langle m \rangle} \log^6(KT) \log^2{(KT / H^{\langle m\rangle})}\log{n}} \right)\right),$$ using at most $T$ time steps and $O(\log\frac{\log m}{\log K} + \log K)$ rounds.
\end{theorem}

\subsubsection{Subsampling and Its Properties}
\label{sec:sampling}

We need the following technical lemma.

\begin{lemma}\label{lem:best-arm}
Let $V$ be a set of $n$ arms, and $\eta$ be the time horizon.
There exists a collaborative algorithm $\A(V, \delta, \eta)$, and functions
\begin{equation}
\label{eq:h-1}
f(V, \delta, \eta) = C_f \cdot \frac{H^{\langle 1 \rangle}(V)}{K}\cdot \log^3{\eta K} \cdot \log\frac{|V|}{\delta},
\end{equation}
and 
\begin{equation}
\label{eq:h-2}
g(V, \delta) = C_g \cdot \frac{H^{\langle 1\rangle}(V)}{K} \cdot  \log\frac{|V|}{\delta},
\end{equation}
where $C_f$  and $C_g$ are universal constants such that

\begin{enumerate}
    \item $\A$ returns an arm from $V$ or $\bot$ using at most $\eta$ time steps and $O(\log K)$ rounds of communication;
    \item $\Pr[\A(V,\delta, \eta) \not \in \{\Top_1(V), \bot\}] \le \delta$;
    \item If $\eta \ge f(V, \delta, \eta)$, then $\Pr[\A(V, \delta, \eta) = \Top_1(V)] \ge 1 - \delta$;
    \item If $\eta \le g(V, \delta)$, then $\Pr[\A(V, \delta, \eta) = \bot] \ge 1 - \delta$.
\end{enumerate}
\end{lemma}

We will show that the algorithm for fixed-time best arm identification in \cite{TZZ19} (denoted by $\B(V, T)$, where $V$ is the input and $T$ is the time horizon), combined with Algorithm~\ref{alg:verify} (setting $m=1$), satisfies Lemma~\ref{lem:best-arm}.  

We first recall the following result from \cite{TZZ19}.  
\begin{lemma}[\cite{TZZ19}]
\label{lem:TTZ19}
For any $V$ and $T$ and fixed $R$, $\B(V, T)$ uses in at most $T$ time steps and $R$ rounds of communication, and returns the best arm with probability at least 
$$
1 - |V| \exp\left(-\Omega\left(\frac{TK^{(R - 1)/R}}{H^{\langle 1 \rangle}(V) \cdot \log^3 \left(T K\right)} \right)\right).
$$
\end{lemma}

\begin{algorithm}[t]
\caption{$\A(V, \delta, T)$}
\label{alg:best-arm}
Run $\B(V, T/2)$, and let the pair $(\{a\}, \tilde{\theta})$ be the certificate produced by $\B$\;
\KwRet{\verM$(V, 1, \{a\}, \tilde{\theta}, \log({|V|/\delta}), T/2)$.}
\end{algorithm}

In Algorithm~\ref{alg:best-arm} we describe how to augment $\B(V, T)$ with a verification step to construct $\A(V, \delta, T)$.
By Lemma~\ref{lem:TTZ19}, Lemma~\ref{lem:verify}, Lemma~\ref{lem:verify-bad}, and setting $R = \log K$, $\A(V, \delta, T)$ satisfies all the four items of Lemma~\ref{lem:best-arm}.
\medskip

By (\ref{eq:h-1}) and (\ref{eq:h-2}), we have \[\frac{f(V, \delta, \eta)}{g(V, \delta)} \le \frac{C_f}{C_g} \cdot \log^3{{\eta K}}\,.\] Define 
$$\beta \triangleq \frac{C_f}{C_g} \cdot \log^3{{T K}}\,.$$

\begin{algorithm}[t]
\caption{$\A'(I, m, \delta, T)$}
\label{alg:subset-best-arm}
\KwIn{a set of $n$ arms $I$, parameters $m$ and $\delta$, and time horizon $T$.}
\KwOut{the top-$1$ arm in a subset of $I$ obtained by randomly sampling each arm in $I$ with probability $1/m$.}
Sample each element from $I$ independently with probability $1/m$; let $V$ be the sampled subset \label{ln:sample}\;
choose $\tau \in \{T, T/\beta\}$ uniformly at random\;
\KwRet{$\A(V, \delta, \tau)$.}
\end{algorithm}

We now design another algorithm $\A'$ for finding the best arm in a random subset of $I$, using Algorithm~\ref{alg:best-arm} as a subroutine. $\A'$ is described in Algorithm~\ref{alg:subset-best-arm}. The following lemma says that any arm in the top-$m$ arms of $I$ will be returned with a good probability by $\A'$.

\begin{lemma}
\label{lem:subset-best-arm}
For any $I$, $\delta \in \left(0, \frac{1}{2}\right)$ and arm $a \in \Top_m$, if \begin{equation}\label{topf:condt}
    T \ge c_{\A} \cdot \frac{H^{\langle m \rangle}}{\delta K m} \cdot \log^6{(TK)} \cdot \log\frac{n}{\delta}
\end{equation}
for a universal constant $c_{\A}$,
then we have
\begin{equation}\label{topf:freq}
\Pr\left[\A'(I, m, \delta, T) = a\right] \ge \frac{(1 - \delta)^2}{e m}\,.
\end{equation}
\end{lemma}

\begin{proof}
We first note that for each $a \in \Top_m$ and $V \subseteq I$ sampled at Line~\ref{ln:sample} in Algorithm~\ref{alg:subset-best-arm}, if $f(V, \delta, \tau) \le T/\beta$ and $\Top_m \cap V = \{a\}$, then by Item 3 of Lemma~\ref{lem:best-arm} we have $\Pr[\A'(I, m, \delta, T) = a] \ge 1 - \delta$.  

According to the uniform subsampling, we have 
\begin{equation}
\label{eq:i-1}
\Pr[\Top_m \cap V = \{a\}] = \frac{1}{m} \cdot \left(1 - \frac{1}{m}\right)^{m-1} \ge \frac{1}{em}.
\end{equation}
For every $i \in I$, let $X_i = 1$ if $i \in V$ and $X_i = 0$ otherwise.  For any $V$ such that $V \cap \Top_m = \{a\}$, we have
\begin{equation*}
\label{eq:i-2}
   H^{\langle 1 \rangle}(V) \le 2 \sum\limits_{i \in I \backslash \Top_m} (\theta_a - \theta_i )^{-2} X_i  \le 2 \sum\limits_{i \in I \backslash \Top_m} {X_i \left(\Delta^{\langle m\rangle}_i\right)^{-2}}\,.
\end{equation*}
We thus have 
\begin{equation*}
\bE[H^{\langle 1 \rangle}(V) \mid \Top_m \cap V = \{a\}] \le \frac{2}{m} H^{\langle m \rangle}.
\end{equation*}
By a Markov inequality, we have that conditioned on $\Top_m \cap V = \{a\}$, with probability $(1-\delta)$, 
\begin{equation}
\label{eq:i-3}
H^{\langle 1 \rangle}(V) \le \frac{2}{\delta m} H^{\langle m \rangle}
\end{equation}
and so on
\begin{equation}
    f(V, \delta, \tau) \le C_f  \frac{H^{\langle 1 \rangle}(V)}{K}  \cdot \log^3{{T K}} \cdot \log\frac{n}{\delta} \le C_f \frac{2 H^{\langle m\rangle}}{\delta K m}  \cdot \log^3{{T K}} \cdot \log\frac{n}{\delta} \le \frac{T}{\beta} \le \tau \,,
\end{equation}
under which 
\begin{equation}
\label{eq:i-4}
\A(V, \delta, \tau)  \text{ outputs $a$ with probability } (1 - \delta) 
\end{equation}
according to Item 3 of Lemma~\ref{lem:best-arm}. 

Finally, we have
\begin{eqnarray*}
\Pr[\A'(I, m, \delta, T) = a] &=& \Pr[\A(V, \delta, \tau) = a] \\
&\ge& \Pr[\A(V, \delta, \tau) = a\ |\ H^{\langle 1 \rangle}(V) \le \tau, \Top_m \cap V = \{a\}] \\
&& \ \ \ \ \cdot \Pr[H^{\langle 1 \rangle}(V) \le \tau \ |\ \Top_m \cap V = \{a\}] \cdot \Pr[\Top_m \cap V = \{a\}] \\
&\ge& (1 - \delta) \cdot (1 - \delta) \cdot {1}/{(em)} = {(1-\delta)^2}/{(em)},
\end{eqnarray*}
where the last inequality is due to (\ref{eq:i-1}), (\ref{eq:i-3}) and (\ref{eq:i-4}). 
\end{proof}

The next lemma is critical.  It says that the probability that Algorithm~\ref{alg:subset-best-arm} outputs any arm from $\Top_m$ cannot be significantly smaller than that of outputting any arm from $I \backslash \Top_m$.

\begin{lemma}
\label{lem:coupling}
For any arms $a \not \in \Top_m$, $b \in \Top_m$ and $\delta \in \left(0, \frac{1}{2}\right)$ we have 
\[
\Pr[\A'(I, m, \delta, T) = a] \le \Pr[\A'(I, m, \delta, T) = b] +  \left(\frac{1}{2em} + \frac{2\delta}{m} \right).
\]
\end{lemma}

\begin{proof}
Let us consider a pair of random sets $(U, V)$, where $V$ is formed by picking each of the $n$ arms of $I$ with probability $1/m$, and $U$ is the set of $V$ after exchanging the assignments of $a$ and $b$.  By this construction we have that
\begin{equation}
\label{eq:j-1}
\Top_1(V) = a \Rightarrow \Top_1(U) = b.
\end{equation}
Moreover, it is easy to see that the marginal distribution of $U$ is {\em identical} to that of $V$.

Note that 
$\A'(I, m, \delta, T) = \A(V, \delta, \tau)$. We thus only need to prove
\begin{eqnarray}
\label{eq:j-0}
\Pr[\A(V, \delta, \tau) = a] \le \Pr[\A(V, \delta, \tau) = b] +  \left(\frac{1}{2em} + \frac{2\delta}{m} \right).
\end{eqnarray}

We start from the left hand side.
\begin{eqnarray}
\Pr[\A(V, \delta, \tau) = a] &=& \Pr[(\A(V, \delta, \tau) = a) \wedge (\Top_1(V) = a)] \nonumber \\
&& + \Pr[(\A(V, \delta, \tau) = a) \wedge (\Top_1(V) \neq a)] \nonumber \\
&\le&  \Pr[(\A(V, \delta, \tau) = a) \wedge (\Top_1(V) = a)] \nonumber \\
&&  + \Pr[(a \in V) \wedge \A(V, \delta, \tau) \not\in \{\Top_1(V), \bot\}] \nonumber \\
&\le&   \Pr[(\A(V, \delta, \tau) = a) \wedge (\Top_1(V) = a)] + \frac{\delta}{m}, \label{eq:j-2}
\end{eqnarray}
where the last inequality follows from the fact that $\Pr[a \in V] = 1/m$ and Item 2 of Lemma~\ref{lem:best-arm}.

By (\ref{eq:j-1}) we have
\begin{equation}
\label{eq:j-3}
    \Pr[(\A(U, \delta, \tau) = b) \wedge (\Top_1(V) = a)] \le \Pr[(\A(U, \delta, \tau) = b) \wedge (\Top_1(U) = b)].
\end{equation}

Since the marginal distribution of $U$ and $V$ are identical, we have
\begin{equation}
\label{eq:j-4}
    \Pr[A(U, \delta, \tau) = b \land \Top_1(U) = b] \le \Pr[A(U, \delta, \tau) = b] = \Pr[A(V, \delta, \tau) = b].
\end{equation}

The following claim enables us to connect $\Pr[\A(V, \delta, \tau) = a]$ and $\Pr[\A(V, \delta, \tau) = b]$.
\begin{claim}
\label{cla:connect}
For any arm $a \not\in \Top_m$, $b \in \Top_m$, and $\delta \in (0, \frac{1}{2})$, we have
\begin{equation}
\label{eq:connect}
\Pr[(\A(V, \delta, \tau) = a) \wedge (\Top_1(V) = a)] \le \Pr[(\A(U, \delta, \tau) = b) \wedge (\Top_1(V) = a)] + \frac{1+\delta}{2 e m} .
\end{equation}
\end{claim}
We will prove Claim~\ref{cla:connect} shortly. By (\ref{eq:j-2}), (\ref{eq:connect}),  (\ref{eq:j-3}), (\ref{eq:j-4}) (one for each inequality below; in order), 
\begin{eqnarray*}
\Pr[\A(V, \delta, \tau) = a] &\le& \Pr[(\A(V, \delta, \tau) = a) \wedge (\Top_1(V) = a)] + \frac{\delta}{m} \\
&\le& \Pr[(\A(U, \delta, \tau) = b) \wedge (\Top_1(V) = a)] +\frac{1+\delta}{2 e m} + \frac{\delta}{m}\\
&\le& \Pr[(\A(U, \delta, \tau) = b) \wedge (\Top_1(U) = b)] + \left(\frac{1}{2em} + \frac{2\delta}{m} \right)\\
&\le& \Pr[\A(V, \delta, \tau) = b] + \left(\frac{1}{2em} + \frac{2\delta}{m} \right).\\
\end{eqnarray*}
This proves Inequality~(\ref{eq:j-0}) and gives the lemma.
\end{proof}

\begin{proof}[Proof of Claim~\ref{cla:connect}]
We first note that if $\Top_1(V) = a$, then by our construction of $U$, the fact that $a \not\in \Top_m$ and $b \in \Top_m$, and the monotonicity of functions $f$ and $g$, we have
\begin{equation}
\label{eq:k-0}
f(U, \delta, \tau) \le f(V, \delta, \tau) \text{ and } g(U, \delta) \le g(V, \delta).
\end{equation}

We expand the left hand side of (\ref{eq:connect}) by running over all subsets $S \subseteq I$.
\begin{equation}
\label{eq:k-1}
\Pr[(\A(V, \delta, \tau) = a) \wedge (\Top_1(V) = a)] = \sum\limits_{S : \Top_1(S) = a} \Pr[(\A(V, \delta, \tau) = a) \wedge (V = S)].
\end{equation}

For each $V = S$ where $\Top_1(S) = a$,  we analyze the quantity 
$$
\Psi \triangleq \Pr[(\A(V, \delta, \tau) = a) \wedge (V = S)] - \Pr[(\A(U, \delta, \tau) = b) \wedge (V = S)]
$$ 
in two cases which cover all possible $T$, in both cases we use properties of $\A$ from Lemma~\ref{lem:best-arm}.
\begin{enumerate}
    
    \item $f(U, \delta, T) \le T$: 
    \begin{eqnarray*}
        \Psi &=& \frac{1}{2}\Pr[(\A(V, \delta, T) = a) \land (V = S)] + \frac{1}{2}\Pr[(\A(V, \delta, {T / \beta}) = a )\land (V = S)] \\
        && -\frac{1}{2}\Pr[(\A(U, \delta, T) = b) \land (V = S)] - \frac{1}{2}\Pr[(\A(U, \delta, {T / \beta}) = b )\land (V = S)] \\
        & \le & \frac{1}{2} + \frac{1}{2} - \frac{1 - \delta}{2} - \frac{0}{2} \le \frac{1 + \delta}{2}
    \end{eqnarray*}
    
    
    \item $\frac{T}{\beta} \le g(U, \delta)$:
    
    \begin{eqnarray*}
        \Psi &=& \frac{1}{2}\Pr[(\A(V, \delta, T) = a) \land (V = S)] + \frac{1}{2}\Pr[(\A(V, \delta, {T / \beta}) = a )\land (V = S)] \\
        && -\frac{1}{2}\Pr[(\A(U, \delta, T) = b) \land (V = S)] - \frac{1}{2}\Pr[(\A(U, \delta, {T / \beta}) = b )\land (V = S)] \\
        & \le & \frac{1}{2} + \frac{\delta}{2} - \frac{0}{2} - \frac{0}{2} \le \frac{1 + \delta}{2}
    \end{eqnarray*}
    
\end{enumerate}
Combining the two cases, we have that 
\begin{eqnarray*} 
&&\Pr[(\A(V, \delta, \tau) = a) \wedge (\Top_1(V) = a)] - \Pr[(\A(U, \delta, \tau) = b) \wedge (\Top_1(V) = a)] \\
&\le& \frac{1+\delta}{2}  \sum\limits_{S : \Top_1(S) = a} \Pr[V = S] \le \frac{1 + \delta}{2} \left(1 - \frac{1}{m}\right)^m \frac{1}{m}\\
&\le& \frac{1+\delta}{2} \cdot \frac{1}{em}\,.
\end{eqnarray*}
\end{proof}

\subsubsection{Reduction to $O(m)$ Arms}
With Lemma~\ref{lem:coupling} we are able to design an algorithm such that given a time horizon $T$, it either returns a superset of $\Top_m$ of size $O(m)$, or $\perp$.   The algorithm is described in Algorithm~\ref{alg:reduction}. The following lemma characterizes the property of Algorithm~\ref{alg:reduction}.

\begin{algorithm}[t]
\caption{\reduct$(I, m, \delta, \gamma, T)$}
\label{alg:reduction}
\KwIn{a set of $n$ arms $I$, parameters $m$, $\delta$, and $\gamma$, and time horizon $T$.}
\KwOut{a super set of $\Top_m$ arms or $\perp$.}
Run independently $z = 25 m \cdot 4^{\gamma} / \delta ^ 2$ copies of $\A'\left(I, m, \delta, T/z\right)$\;
let $\hat{\alpha}_i$ be the frequency of $i$-th arm among the $z$ returned values by $\A'$\;
\If{($m$-th largest $\hat{\alpha}_i$) $< {3}/{(4em)}$}{\KwRet{$\perp$}\;}
\Else{\KwRet{all arms $i$ for which $\hat{\alpha}_i \ge {1}/{(16em)}$}.}
\end{algorithm}

\begin{lemma}
\label{lem:reduction}
For any $I$, $m$, $\delta \in \left(0, \frac{1}{24}\right)$, and $\gamma$, we have
$$
\Pr[(\Top_m \subseteq \text{\reduct}(I, m, \delta, \gamma, T)) \vee (\text{\reduct}(I, m, \delta, \gamma, T) = \perp)] \ge 1 - n \cdot \exp(-4^{\gamma}).
$$
Moreover, If \[T \ge c_R  \cdot 4^{\gamma} \cdot \frac{H^{\langle m \rangle}}{\delta^3 K } \cdot \log^6 TK \cdot \log\frac{n}{\delta}\] for a universal constant $c_R$, then 
$$
\Pr[\Top_m \subseteq \text{\reduct}(I, m, \delta, \gamma, T)] \ge 1 - n \cdot \exp(-4^{\gamma}).
$$
Finally, if $\text{\reduct}(I, m, \delta, \gamma, T) \neq \perp$, then the number of returned arms is bounded by $O(m)$.
\end{lemma}

\begin{proof}
Let $\alpha_i$ be the probability that $\A'$ returns the $i$-th arm, that is, 
$$
\alpha_i \triangleq \Pr\left[\A'\left(I, m, \delta, \frac{T}{z}\right) = i\right].
$$
Let $H = \{i \in I\ |\ \alpha_i \ge \frac{5}{8em}\}$.  By Chernoff-Hoeffding we have that
\begin{equation*}
\forall i \in I \backslash H:  \Pr\left[\hat{\alpha}_i \ge (1 + \delta) \cdot \frac{5}{8 e m}\right] \leq \exp\left(-\frac{5 \delta^2 z}{24 e m}\right) \le \exp( -4^{\gamma}).
\end{equation*}
Since $\frac{5(1 + \delta)}{8em} < \frac{3}{4em}$, it holds that 
\begin{equation}
\label{eq:l-1}
\Pr\left[\forall{i \in I \setminus H} : 
\hat{\alpha}_i < \frac{3}{4em}\right] \ge 1 - n \cdot \exp(-4^{\gamma}).
\end{equation}

For the first part of the lemma, we consider two cases regarding $H$.
\begin{enumerate}
\item $\abs{H} < m$: By (\ref{eq:l-1}), Algorithm~\ref{alg:reduction} returns $\perp$ with probability at least $1 - n \cdot \exp(-4^{\gamma})$.

\item $\abs{H} \ge m$: We must have $\Top_m \subseteq H$ or $H \setminus \Top_m \neq \varnothing$. In the first case, by the definition of $H$ we have that $\forall i \in \Top_m : \alpha_i \ge \frac{5}{8 e m} \ge \frac{1}{8 e m}$. In the second case, by Lemma~\ref{lem:coupling} we have \[\forall{i \in \Top_m} : \alpha_i \ge \frac{3}{4em} - \left(\frac{1}{2em} + \frac{2\delta}{m}\right) \ge \frac{1}{8em}\,.\] By Chernoff-Hoeffding we have
\begin{equation*}
\forall{i \in \Top_m} : \Pr\left[\hat{\alpha}_i < \frac{1}{16em}\right] \le \Pr\left[\hat{\alpha}_i < (1 - \delta) \cdot \frac{1}{8em}\right] \le \exp\left(-\frac{\delta^2 z}{8em}\right) \le \exp(-4^{\gamma}).
\end{equation*}
Therefore with probability at least $1 - n \cdot \exp(-4^\gamma)$, Algorithm~\ref{alg:reduction} outputs a super set of $\Top_m$.
\end{enumerate}
Combining the two cases, with probability at least $1 - n \cdot \exp(-4^\gamma)$, Algorithm~\ref{alg:reduction} outputs either $\perp$ or  a super set of $\Top_m$.

For the second part, since \[\frac{T}{z} \ge  c_{T} \cdot \frac{H^{\langle m\rangle}}{25 \delta K m}\log^6(TK) \log\frac{n}{\delta}  \ge c_{\A} \cdot \frac{H^{\langle m\rangle}}{\delta K m} \log^6(TK) \log\frac{n}{\delta}\] for a large enough constant $c_T$, by Lemma~\ref{lem:subset-best-arm} we have $\alpha_i \ge \frac{(1-\delta)^2}{em}$ for any $i \in \Top_m$.  By Chernoff-Hoeffding we have
\begin{equation}
    \forall{i \in \Top_m} : \Pr\left[\hat{\alpha}_i \le \frac{(1 - \delta)^3}{em}\right] \le \exp\left( - \frac{\delta^2 (1 - \delta)^2 z}{em} \right) \le \exp(-4^{\gamma}).
\end{equation}
Since $\frac{(1 - \delta)^3}{e m} > \frac{3}{4 e m}$, Algorithm~\ref{alg:reduction} outputs a superset of $\Top_m$ with probability at least $1 - n \cdot \exp(-4^{\gamma})$.

Finally, since $\sum_{i \in I} \hat{\alpha}_i \le 1$, there are at most $O(m)$ arms $i$ with $\hat{\alpha}_i \ge 1/(16em)$.
\end{proof}

We now use Algorithm~\ref{alg:reduction} as a building block for designing the reduction algorithm for general time horizon.  The final reduction is described in Algorithm~\ref{alg:reduction-general}.

\begin{algorithm}[t]
\caption{\text{\reductG}$(I, m, T)$}
\label{alg:reduction-general}
\KwIn{a set of $n$ arms $I$, parameter $m$, and time horizon $T$.}
\KwOut{a super set of top-$m$ arms.}
\For{$s = 1, 2, \dotsc $}{
$B_s \gets \text{\reduct}\left(I, m, \frac{1}{25}, s, \frac{6 T}{\pi^2 s^2}\right)$\;
}
let $s^{\star}$ be the largest $s$ such that $B_s \neq \bot$\;
\If{there is no such $s$}{\KwRet{an arbitrary set of $m$ arms}\;}
\Else{\KwRet{$B_{s^{\star}}$}.}
\end{algorithm}

The following lemma summarizes the property of Algorithm~\ref{alg:reduction-general}.  

\begin{lemma}
\label{lemma:reduction-general}
For any $I, m$ and $T$, the collaborative algorithm \text{\reductG}$(I, m, T)$ returns a super set of $\Top_m$ of size $O(m)$ with probability at least 
$$
1 - 2 n \cdot \exp\left(-\Omega\left(\frac{KT}{H^{\langle m\rangle} \log^6(TK) \log^2(TK / H^{\langle m\rangle}) \log{n}}\right)\right).
$$
\end{lemma}

\begin{proof}
Let $s'$ be the largest $s$ such that $\frac{6 T}{\pi^2 s^2} \ge 4^{s} \cdot c_R \frac{H^{\langle m \rangle}}{\delta^3 K} \log^6{TK} \log\frac{n}{\delta}$ with $\delta = \frac{1}{25}$.  If there is no such $s$ then the lemma follows trivially. 
Otherwise we have
\begin{equation}
\label{eq:m-0}
4^{s'} = \Omega\left(\frac{KT}{H^{\langle m\rangle} \log^{6}(KT) \log^2{(TK / H^{\langle m\rangle})}\log{n} }\right).
\end{equation}

By Lemma~\ref{lem:reduction} we have
\begin{equation}
\label{eq:m-1}
    \Pr\left[\Top_m \subseteq B_{s'}\right] \ge 1 - n \cdot \exp(-4^{s'}).
\end{equation}
For any $s = s' + j\ (j \ge 1)$, by Lemma~\ref{lem:reduction} we have
\begin{equation}\label{eq:m-2}
\Pr\left[(\Top_m \subseteq B_s) \vee (B_s = \perp) \right] \ge 1 - n \cdot \exp\left(-4^s\right). 
\end{equation}
By (\ref{eq:m-1}) and (\ref{eq:m-2}) we have 
\begin{eqnarray*}
    \Pr[\Top_m \subseteq B_{s^{\star}}] &\ge& \Pr[\Top_m \subseteq B_{s'}] - \sum\limits_{s = s' + 1}^{\infty}\Pr[(\Top_m \not\subseteq B_s) \wedge (B_s \neq \perp)] \\
    &\ge& 1 - n \cdot \exp(-4^{s'}) - n \cdot \sum\limits_{s=s'+1}^\infty \exp(-4^{s}) \\
    &\ge& 1 - 2n \cdot \exp(-4^{s'}),
\end{eqnarray*}
which, combined with (\ref{eq:m-0}), concludes the lemma.  
\end{proof}

Now we are ready to prove Theorem~\ref{thm:ub-improve}.

\begin{proof}[Proof of Theorem~\ref{thm:ub-improve}]
We describe the algorithm in Algorithm~\ref{alg:improve}.  We first use Algorithm~\ref{alg:reduction-general} 
to reduce the number of arms to $O(m)$, and then call Algorithm~\ref{alg:general} \colMG\ to compute the set of top-$m$ arms. 

The success probability of Theorem~\ref{thm:ub-improve} follows directly from Lemma~\ref{lemma:reduction-general} and Theorem~\ref{thm:ub-general}.  

The round complexity follows from Lemma~\ref{lem:best-arm}, Theorem~\ref{thm:ub-general}, and the ways Algorithm~\ref{alg:reduction-general}, Algorithm~\ref{alg:reduction}, and Algorithm~\ref{alg:subset-best-arm} are designed (they can be perfectly parallelized). Basically, the first term $O\left(\log \frac{\log m}{\log K}\right)$ comes from Algorithm~\ref{alg:reduction} in which we reduce the number of arms from $O(m)$ to $O(K^{10})$.  The second term $O(\log K)$ comes from two sources.  One is from Line~\ref{ln:top-1} of Algorithm~\ref{alg:improve} where we have used the best arm identification algorithm in \cite{TZZ19} (setting $R = \log K$), and the other is due to the run of \colN\ on the remaining $K^{10}$ arms (setting $R = \log n = \log K^{10}$).   The round complexities introduced by other steps are negligible.
\end{proof}

\begin{algorithm}[t]
\caption{\colMI$(I, m, T)$}
\label{alg:improve}
\KwIn{a set of $n$ arms $I$, parameter $m$, and time horizon $T$.}
\KwOut{the set of top-$m$ arms of $I$.}
$B \gets$ \text{\reductG}$(I, m, T / 2)$\;
\KwRet{\colMG$(B, m, T / 2)$.\label{ln:top-1}}
\end{algorithm}

We comment that we can make the the second term in the round complexity of Theorem~\ref{thm:ub-improve} to be an arbitrary number $R$, at the cost of slightly reducing the speedup. These can be accomplished by using the general round-speedup of Lemma~\ref{lem:TTZ19} and Lemma~\ref{lem:colN}, by setting round complexity to $R$ and $10 R$ respectively. On the other hand, the first term $O\left(\log \frac{\log m}{\log K}\right)$  remain the same even we target a $\tilde{\Theta}(\sqrt{K})$ speedup. We will show in Section~\ref{sec:lb} that this is inevitable. 

\begin{theorem}
\label{thm:ub-final}
Let $I$ be a set of $n$ arms, and $m \in \{1, \ldots, n - 1\}$. Let $T$ be a time horizon and $R$ be a parameter $(1 \le R \le \log K)$.
There exists a collaborative algorithm that computes the set of top-$m$ arms of $I$ with probability at least 
$$1 - n \cdot \exp\left(-\Omega\left(\frac{TK^{(R - 1)/R}}{H^{\langle m\rangle}\cdot (\log({H^{\langle m\rangle} K}) + R \log{n})\cdot \log\log{n} \cdot \log^2\left(TK/H^{\langle m\rangle}\right)}\right)\right)$$
using at most $T$ time steps and $O\left(\log\frac{\log m}{\log K} + R\right)$ rounds.
\end{theorem}

If we want to present Theorem~\ref{thm:ub-final} in the form of round-speedup tradeoff, we have the following corollary. 

\begin{corollary}\label{cor:tradeoff}
There is a collaborative algorithm for the top-$m$ arms identification problem that achieves $\tilde{\Omega}(K^{(R - 1) / R})$ speedup using at most $O(\log \frac{\log m}{\log K} + R)$ rounds of communication.
\end{corollary}

\subsection{Auxiliary Algorithms}
\label{sec:ub-aux}

\subsubsection{PAC Top-$m$ Arm Identification} 
\label{sec:PAC}

Let $\beta(u, t) \coloneqq \sqrt{\frac{1}{2 u}\log\left(\frac{5n t^4 }{4 \delta} \right)}$.  We first recall the \lucb\ algorithm from \cite{KTAS12}, which is described in Algorithm~\ref{alg:lucb}.

\begin{algorithm}[t]
\caption{\lucb$(I, m, \eps, \delta)$}
\label{alg:lucb}
\KwIn{a set of $n$ arms $I$, parameter $m$, parameters $\eps, \delta \in (0,1)$}
\KwOut{the set of $m$ arms such that each arm is $(\eps, m)$-top}
Pull each arm once; $\forall i, p_i \gets 1$\;
\For{$t = |I| + 1, |I| + 2, \dotsc$}{
Let $H$ be the set of $m$ arms with highest estimated mean $\hat{\theta}_i$, and $L \gets I \setminus H$\;
$h^{\star} \gets \arg\min_{i \in H}\{\hat{\theta}_i - \beta(p_i, t)\}$, and 
$l^{\star} \gets \arg\max_{i \in L}\{\hat{\theta}_i + \beta(p_i, t)\}$\;
\If{$(\hat{\theta}_{l^{\star}} + \beta(p_{l^{\star}}, t)) - (\hat{\theta}_{h^{\star}} - \beta(p_{h^{\star}}, t)) < \eps / 2$}{
    \KwRet{$H$\;}
}
pull $h^{\star}$ and $l^{\star}$, and increment $p_{l^{\star}}$ and $p_{h^{\star}}$.
}
\end{algorithm}

The following lemma summarizes the properties of Algorithm~\ref{alg:lucb}.

\begin{lemma}[\cite{KTAS12}]
\label{lem:lucb}
Algorithm~\ref{alg:lucb} \lucb$(I, m, \eps, \delta)$ returns the set of $(\eps, m)$-top arms with probability at least $(1 - \delta)$ using at most $O\left(H^{\langle m\rangle}_{\eps} \log \left(H^{\langle m \rangle}_{\eps} / \delta\right)\right)$ time steps. Moreover,
$$
\Pr[\forall{i \in I} : |\hat{\theta}_i - \theta_i| \le \max\{\Delta^{\langle m \rangle}_i / 4, \eps / 4\}] \ge 1 - \delta.
$$
\end{lemma}

Note that Algorithm~\ref{alg:lucb} is used to minimize the time (number of pulls) given PAC parameters $\eps$ and $\delta$.  In our task we need to (implicitly) minimize $\eps$ given $T$ and $\delta$. For this purpose we need to modify Algorithm~\ref{alg:lucb} a bit; we described our new algorithm in Algorithm~\ref{alg:cenT}.

\begin{algorithm}[t]
\caption{\cenT$(I, m, T, \delta)$}
\label{alg:cenT}
\KwIn{a set of $n$ arms $I$, parameter $m$, time horizon $T$, and parameter $\delta$.}
\KwOut{the set of $m$ arms such that each arm is $(\eps, m)$-top.}
Pull each arm once; $\forall i, p_i \gets 1$ and set set $t \gets n$\;

\While{$t \le T - 2$}{
let $H$ be the set of $m$ arms with highest estimated mean $\hat{\theta}_i$, and $L \gets I \setminus H$\;
$h^{\star} \gets \arg\min_{i \in H}\{\hat{\theta}_i - \beta(p_i, t)\}$, and 
$l^{\star} \gets \arg\max_{i \in L}\{\hat{\theta}_i + \beta(p_i, t)\}$\;
pull $h^{\star}$ and $l^{\star}$, update $\hat{\theta}_{l^{\star}}$ and $\hat{\theta}_{l^{\star}}$, set $t \gets t + 2$, and increment $p_{l^{\star}}$ and $p_{h^{\star}}$\;
}
\KwRet{H}.
\end{algorithm}

The following lemma summarizes the properties of Algorithm~\ref{alg:cenT}.  It can be proved in essentially the same way as that for Lemma~\ref{lem:lucb} in \cite{KTAS12}.

\begin{lemma}
\label{lem:cenT}

If $T \geq c_1 H^{\langle m\rangle}_{\eps} \log\left(H^{\langle m \rangle}_{\eps} / \delta\right)$ for a large enough constant $c_1$, then with probability at least $(1 - \delta)$, \cenT$(I, m, T, \delta)$ returns a set of $(\eps, m)$-top arms.  Moreover, 
$$
\Pr[\forall{i \in I} : |\hat{\theta}_i - \theta_i| \le \max\{\Delta^{\langle m \rangle}_i / 4, \eps / 4\}] \ge 1 - \delta.
$$
\end{lemma}

The algorithm \cenB\ is almost identical to \cenT: we can just follow \cenT, but replace all the sample values $x$ with $1 - x$.  

\begin{lemma}
\label{lem:cenB}

If $T \geq c_1 H^{\langle n - m\rangle}_{\eps} \log\left(H^{\langle n - m \rangle}_{\eps} / \delta\right)$ for a large enough constant $c_1$, then with probability at least $(1 - \delta)$, \cenB$(I, m, T, \delta)$ returns a set of $(\eps, m)$-bottom arms and
$$
\Pr[\forall{i \in I} : |\hat{\theta}_i - \theta_i| \le \max\{\Delta^{\langle n - m \rangle}_i / 4, \eps / 4\}] \ge 1 - \delta.
$$
\end{lemma}

\subsubsection{A Simple Collaborative Algorithm for Top-$m$ Identification}
\label{sec:simple-alg}

The \colN\ algorithm, described in Algorithm~\ref{alg:colN}, is a slightly modified version of the {\em successive accepts and rejects} (SAR) algorithm in \cite{ABM10}. The goal of the modification is to achieve a small number of rounds of communication in the collaborative setting.

\begin{algorithm}[t]
\caption{\colN$(I, m, T)$}
\label{alg:colN}
\KwIn{a set of arms $I$, parameter $m$, and time horizon $T$.}
\KwOut{the set of top-$m$ arms.}
$I_0 \gets I$, $m_0 \gets m$, $\Acc_1 \gets \emptyset$, $\Theta \gets \emptyset$\; 
$T_0 \gets 0$, $T_r \gets \left\lfloor\frac{n^{r / R} T}{n^{1+1/R} (R + 1)}\right\rfloor$ for $r = 1, \ldots, R + 1$\;
$n_r \gets \left\lfloor\frac{n}{n^{r/R}}\right\rfloor$ for $r = 0, \ldots, R + 1$\;
\For{
$r = 0, \dotsc, R$
}{
each agent pulls each arm $i \in I_r$ for $(T_{r+1} - T_{r})$ times\; let $\hat{\theta}^{(r)}_i$ for $i \in I_r$ be the aggregated mean of $i$-th arm  after $K T_{r}$ pulls\;
let $\sigma_r : \{1, \dotsc, |I_r|\} \to I_r$ be the bijection such that $\hat{\theta}_{\sigma_r(1)}^{(r)} \ge \hat{\theta}_{\sigma_r(2)}^{(r)} \ge \dotsc \ge \hat{\theta}_{\sigma_r(|I_r|)}^{(r)}$\;
for $i \in I_r$ define empirical gaps $\Delta^{(r)}_i = 
\begin{cases}\hat{\theta}_i^{(r)} - \hat{\theta}_{\sigma_r(m_r + 1)}^{(r)}, & \text{if $\hat{\theta}_i^{(r)} \ge \hat{\theta}_{\sigma_r(m_r)}^{(r)}$,} \\
\hat{\theta}_{\sigma_r(m_r)}^{(r)} - \hat{\theta}_i^{(r)}, & \text{if $\hat{\theta}_i^{(r)} \le \hat{\theta}_{\sigma_r(m_r + 1)}^{(r)}$}
\end{cases}\,;$ \\
let $E_r$ be the set of $(n_r - n_{r+1})$ arms from $I_r$ with the largest gaps $\Delta^{(r)}_i$\;
$\Acc_{r + 1}\gets \Acc_r \cup \left\{i \in E_r \mid \hat{\theta}_i^{(r)} \ge \hat{\theta}_{\sigma_r(m_r)}^{(r)}\right\}$\;
for $i \in E_r$ add $\hat{\theta}^{(r)}_i$ to $\Theta$\;
set $I_{r + 1} \gets I_r \setminus E_r$ and $m_{r + 1} \gets m - \abs{\Acc_{r + 1}}$\;
}
\KwRet{$(\Acc_{R+1}, \Theta)$}.
\end{algorithm}

\begin{lemma}
\label{lem:colN}
For any fixed $R$, Algorithm~\ref{alg:colN} \colN$(I, m, T)$ uses at most $T$ time steps and $R + 1$ rounds of communication, and returns a top-$m$ certificate $(S, \Theta)$ with probability at least
$$1 - 2n (R + 1) \cdot \exp\left(-\frac{KT}{256  \cdot n^{1/R} H^{\langle m \rangle}\cdot (R + 1)\log 2 n }\right).$$
\end{lemma}

\begin{proof}
The $R + 1$ round complexity is clear from the description of the algorithm. The running time can be upper bounded by
$$
\sum\limits_{r=0}^{R} n_r \cdot T_{r + 1} \le \sum_{r=0}^{R}\frac{n}{n^{r/R}} \cdot \frac{n^{r/R}T}{n(R + 1)} = T.
$$

We next bound the error probability.  The argument is very similar to the one in \cite{ABM10}, and we include it here for completeness.

Let $\pi: \{1, \dotsc, n\} \to I$ be the bijection such that 
$$
\Delta^{\langle m \rangle}_{\pi(1)} \le \Delta^{\langle m\rangle}_{\pi(2)} \le \dotsc \le \Delta^{\langle m\rangle}_{\pi(n)}\,.
$$
Consider the event 
$$
\E_2 : \left\{\forall i \in I, \forall r = 0, \dotsc, R: 
|\hat{\theta}^{(r)}_i - \theta_i| \le \Delta^{\langle m\rangle}_{\pi(n_{r+1} )} / 8
\right\}.
$$
We have
\begin{eqnarray}
\Pr\left[\bar{\E}_2\right] &\le& \sum_{i \in I}\sum_{r = 0}^{R}\Pr\left[|\hat{\theta}^{(r)}_i - \theta_i| > \Delta^{\langle m\rangle}_{\pi(n_{r+1})} / 8\right] \label{eq:o-1} \\
&\le& \sum_{i \in I}\sum_{r = 0}^{R} 2\exp\left(-KT_{r+1} \cdot \left(\Delta^{\langle m \rangle}_{\pi(n_{r+1} )}/8\right)^2\right) \label{eq:o-2} \\
&\le & 2n (R + 1) \cdot \exp\left(-\frac{K T}{256 \cdot n^{1/R} H^{\langle m\rangle}\cdot(R + 1)\log 2n}\right), \label{eq:o-3}
\end{eqnarray}
where (\ref{eq:o-1}) $\to$ (\ref{eq:o-2}) is due to Chernoff-Hoeffding, and (\ref{eq:o-2}) $\to$ (\ref{eq:o-3}) we have used an inequality from \cite{ABM10}:
$$\max_{i \in I} \left\{{i}/{\left(\Delta^{\langle  m\rangle}_{\pi(i)}\right)^{2}}\right\} \le \log 2n \cdot H^{\langle m \rangle}.$$

Once $\E_2$ holds, the proof for that fact that the returned pair $(S, \Theta)$ is a top-$m$ certificate is straightforward.
\end{proof}

\section{Lower Bounds for the Fixed-Time Case}
\label{sec:lb}

In this section, we prove the following lower bound theorems for the fixed-time setting.

\begin{theorem}\label{thm:lb-logK}
For every $K$, $m$ ($m \leq K$), and $\alpha$ ($\alpha \in [1, K^{0.1}]$), if a fixed-time collaborative algorithm $\mathcal{A}$ with $K$ agents returns the top-$m$ arms for every instance $J$  with probability at least $0.99$, when given time budget $\frac{\alpha}{17K} \cdot H^{\langle m \rangle}(J)  $, then there exists an instance $J'$  such that $\mathcal{A}$ uses $\Omega(\ln K / (\ln \ln K  + \ln \alpha))$ rounds of communication  in expectation given instance $J'$ and time budget $ \frac{\alpha}{17 K} \cdot H^{\langle m \rangle}(J')$. 

In other words, to achieve $(K / \alpha)$ speedup for identifying the top $m$ arms, the collaborative algorithm needs $\Omega(\ln K / (\ln \ln K  + \ln \alpha))$ communication rounds.
\end{theorem}

We will prove Theorem~\ref{thm:lb-logK} in Section~\ref{sec:lb-logK}. It is relatively easy and resembles the round complexity lower bound $\Omega(\ln K / (\ln \ln K  + \ln \alpha))$ for top arm identification in the fixed-time setting \cite{TZZ19}.

\begin{theorem}\label{thm:lb-loglogm}
For every large enough $K$ and $m$ such that $K \geq \Omega(\ln^4 m)$, if a fixed-time collaborative algorithm $\mathcal{A}$ with $K$ agents returns the top-$m$ arms for every instance $J$ with probability at least $0.99$, when given time budget $\frac{1}{\sqrt{K}} \cdot H^{\langle m \rangle}(J)$, then there exists an instance $J'$ such that $\mathcal{A}$ uses $\Omega(\ln (\ln m / \ln K))$ rounds of communication given instance $J'$ and time budget $\frac{1}{\sqrt{K}} \cdot H^{\langle m \rangle}(J')$.

In other words, even if one only aims at $\sqrt{K}$ speedup, the collaborative algorithm needs \[\Omega(\ln (\ln m / \ln K))\] rounds of communication.
\end{theorem}

Theorem~\ref{thm:lb-loglogm} will be proved in Section~\ref{sec:lb-loglogm}. It marks the different round complexity requirement for collaborative multiple arm identification compared to the best arm identification problem. It is known that only constant number of round is needed to achieve $0.99$ success probability using $\tilde{O}(K^{-\zeta} \cdot H^{\langle m \rangle}(J))$ time budget (i.e., $\tilde{O}(K^{\zeta})$ speedup) for every constant $\zeta \in (0, 1)$ \cite{HKKLS13,TZZ19}. However, Theorem~\ref{thm:lb-loglogm} rules out such possibility for the multiple arm identification problem, proving it much harder than best arm identification in the collaborative setting. We note that we only prove the lower bound for $\zeta = 1/2$, for the simplicity of the exposition. However, the proof can be easily extended to any constant $\zeta > 0$. The only differences are that, in the theorem statement, the constraint $K \geq \Omega(\ln^4 m)$  will become $K \geq \ln^{f(\zeta)} m$, and the round complexity lower bound will become $\frac{1}{f(\zeta)} \cdot \ln (\ln m / \ln K)$ where $f(\zeta) > 0$ increases as $\zeta$ approaches $0$.

\subsection{Proof of Theorem~\ref{thm:lb-logK}} \label{sec:lb-logK}

The proof of Theorem~\ref{thm:lb-logK} is via a simple reduction from the following lower bound for the best arm identification problem.

\begin{theorem}[Theorem 10 in \cite{TZZ19}]\label{thm:lb-best-arm}
For every $K$ and $\alpha$ ($\alpha \in [1, K^{0.1}]$), if a fixed-time collaborative algorithm $\mathcal{B}$ with $K$ agents returns the best arm for every instance $I$ with probability at least $0.99$, when given time budget $ \frac{\alpha}{K} \cdot H^{\langle 1 \rangle}(I)$, then there exists an instance $I'$ with $H^{\langle 1 \rangle}(I') \geq K$ such that 1) the mean reward of the best arm in $I'$ is less or equal to $1/2$ and 2) $\mathcal{B}$ uses $\Omega(\ln K / (\ln \ln K  + \ln \alpha))$ rounds of communication  in expectation given instance $I'$ and time budget $ \frac{\alpha}{K} \cdot H^{\langle 1 \rangle}(I') $.
\end{theorem}

\begin{proof}[Proof of Theorem~\ref{thm:lb-logK}]
Given an algorithm $\mathcal{A}$ to identify the top-$m$ arms, we construct an algorithm $\mathcal{B}$ to identify the best arm as follows.

For every instance $I$ where the mean reward of the best arm is at most $1/2$, we create $(m-1)$ artificial arms with mean reward $3/4$, and together with $I$ we have an instance $J$. We have that $H^{\langle m \rangle}(J) \leq 16 (m-1) + H^{\langle 1 \rangle}(I) $. When $ H^{\langle 1 \rangle}(I) \geq K \geq m$, we further have $H^{\langle m \rangle}(J)\leq 17  H^{\langle 1 \rangle}(I)$. If $\mathcal{A}$ is given time budget $T$, the algorithm $\mathcal{B}$ will simulate $\mathcal{A}$ with instance $J$ and time budget $T$.
Whenever an artificial arm is queried by $\mathcal{A}$, $\mathcal{B}$ will generate a sample from a Bernoulli variable with mean $3/4$; otherwise, $\mathcal{B}$ queries the real arm in $I$ and feed the observation back to $\mathcal{A}$. The total time used by $\mathcal{B}$ (only counting queries to the real arms) is at most $T$, satisfying the time budget constraint. When $\mathcal{A}$ returns the identified arms, $\mathcal{B}$ will return any real arm from the set (or declare failure if no such arm exists). If $\mathcal{A}$ returns the correct top $m$ arms in $J$ (which happens with probability at least $0.99$), $\mathcal{B}$ will also return the correct best arm of $I$. 

Now invoking Theorem~\ref{thm:lb-best-arm}, we know that there exists an instance $I'$ with $ H^{\langle 1 \rangle}(I') \geq K$ such that $\mathcal{B}$ uses $\Omega(\ln K / (\ln \ln K + \ln \alpha))$ rounds of communication in expectation for instance $I'$ and time budget $T = \frac{\alpha}{K} \cdot H^{\langle 1 \rangle}(I')$. Therefore, we have that $\mathcal{A}$ uses $\Omega(\ln K  / (\ln \ln K + \ln \alpha))$ rounds of communication in expectation for the corresponding instance $J'$ (constructed from $I'$) and time budget $T$. Also note that $T \geq \frac{\alpha}{ 17 K} H^{\langle m \rangle}(J')$, and we complete the proof.
\end{proof}

\begin{remark}
Given the above proof, one may observe that the $m \leq K$ constraint in Theorem~\ref{thm:lb-logK} can be relaxed to admit much bigger $m$ if Theorem~\ref{thm:lb-best-arm} can be strengthened to provide lower bound instances $I'$ such that $H^{\langle 1 \rangle}(I') \gg K$. We highly believe this is possible, and we do have a proof sketch in the core of which is an improved argument of Section 3.2 in \cite{TZZ19}. Since the proof is quite involved and not directly related to this paper, we leave the formal proof as a future work.
\end{remark}

\subsection{Proof of Theorem~\ref{thm:lb-loglogm}} \label{sec:lb-loglogm}

\subsubsection{The Hard Instances and the Proof Intuition}\label{sec:hard-instance}

For any fixed $K$ (the number of agents), we define a  distribution of instances $\I(C, \mu, n)$ with $n$ Bernoulli arms, where we assume $n$ is an odd integer, the parameter $C \in (0, 1/4)$ denotes the gap between the mean rewards of the top arm and the bottom arm, the parameter $\mu \in (3/8, 5/8)$. We always set $m = (n-1)/2$, i.e., the goal is to identify the top half arms (not including the median arm). 

When $n \leq K^{10}$, we set $\I(C, \mu, n)$ to be a deterministic instance where the $(n-1)/2$ top arms have mean reward $(\mu + C/2)$, and the $(n-1)/2$ bottom arms have mean reward $(\mu - C/2)$. There is $1$ middle arm with mean reward $\mu$ sandwiched between the top and the bottom arms.  

When $n > K^{10}$, we define a random sample $I \sim \I(C, \mu, n)$ in a recursive fashion as follows (and illustrated in Figure~\ref{fig:lb-I}). Let $\eta$ be the smallest odd integer that is greater than $\sqrt[4]{n}$. There are $((n - \eta(2\eta + 1))/2  + \xi \eta)$ top arms with mean reward $(\mu + C/2)$, and there are $((n - \eta(2\eta+1))/2  - \xi \eta)$ bottom arms with mean reward $(\mu - C/2)$, where the \emph{bias} $\xi$ is an integer independently and uniformly sampled from $[-\eta, \eta]$. For the rest $\eta(2\eta+1)$ arms in the middle, we make $(2\eta+1)$ independent samples $I_1, I_2, \dots, I_{2\eta+1}$ (each of which has $\eta$ arms), such that for each $j \in [2\eta + 1]$, we have
\[
I_j \sim \I\left(C\sqrt{\eta/n}, \mu + \frac{j - \eta - 1}{8} \cdot C n^{-1/4}, \eta\right). 
\]
The final instance $I$ consists of the union of the arms in $I_j$ ($j \in [2\eta + 1]$) together with the top and the bottom arms.
We also say that the arm is in the $j$-th \emph{block} if it is an arm in $I_j$.

\begin{figure}[t!]
\begin{center}
    \includegraphics[width=0.8\textwidth,clip,trim=0 10.5cm 0 10cm]{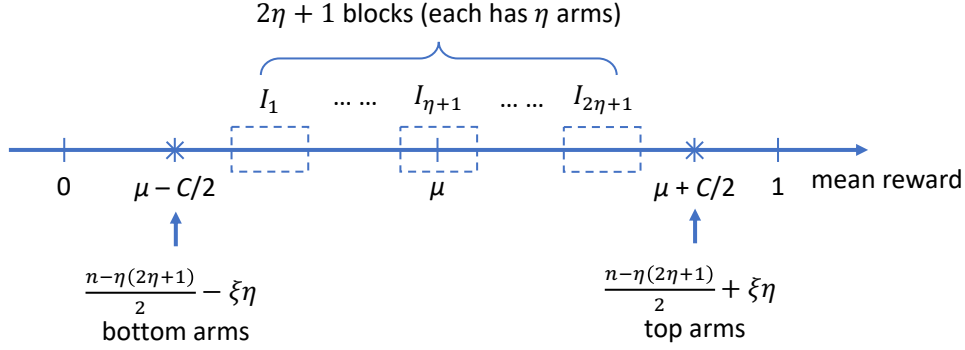}
\end{center}
\caption{Illustration of the mean rewards of the arms in $I \sim \mathcal{I}(C, \mu, n)$ for $n > K^{10}$.} \label{fig:lb-I}
\end{figure}

Below we will claim a few properties about our constructed hard instances. First, it is straightforward to verify the lemma.
\begin{lemma}\label{lem:lb-instance-fact}
For sufficiently large $n$, we have the following claims.
\begin{enumerate}
    \item For any arm in the $j$-th block, its mean reward $\theta \in \frac{1}{2} +  Cn^{-1/4} \cdot (\frac{j-\eta-1}{8} \pm \frac{1}{100})$. Therefore, the mean rewards of all middle arms are sandwiched between the top and the bottom arms, and any two distinct blocks do not overlap.
    \item The median arm of $I$ is the median arm of $I_{\xi + \eta + 1}$.  \label{lem:lb-instance-fact-item2}
\end{enumerate}
\end{lemma}

In the following lemma, we show the order of the complexity measure of the constructed instances.

\begin{lemma}\label{lemma:lb-H}
For each instance $I$ in the support of $\I(C, \mu, n)$, we have $H^{\langle m \rangle}(I) = \Theta(C^{-2} n \cdot \ln (e + \frac{\ln n}{\ln K}))$.
\end{lemma}
\begin{proof}
We prove this lemma via induction, where the base case $n \leq K^{10}$ is straightforward to verify.

When $n > K^{10}$, let $\eta$ be defined in the construction of the instances. Let $I_1, I_2, \dots, I_{2\eta+1}$ be any instances such that $I_j$ is in the support of $\I(C\sqrt{\eta/n}, \mu + \frac{j - \eta - 1}{8} \cdot C n^{-1/4}, \eta)$ for each $j \in [2\eta+1]$. Let $\xi \in [-\eta, \eta]$ be any integer, and let $I$ be the instance constructed using the parameters above. Let $\theta_{[m]}$ and $\theta_{[m+1]}$ be the mean rewards of the $m$-th and the $(m+1)$-th best arm, respectively. We have
\begin{multline}
H^{\langle m \rangle}(I) = \\ \left(\frac{n - \eta(2\eta + 1)}{2}  + \xi \eta\right) \cdot \left(\frac{1+C}{2} - \theta_{[m+1]}\right)^{-2} + \left(\frac{n - \eta(2\eta + 1)}{2}  - \xi \eta\right) \cdot \left(\frac{1-C}{2} - \theta_{[m]}\right)^{-2}  \\
  + \sum_{j = 1}^{\xi + \eta} \sum_{\text{arm $i$ in block $j$}} (\theta_i - \theta_{[m]})^{-2} + \sum_{j = \xi+\eta+2}^{2\eta + 1} \sum_{\text{arm $i$ in block $j$}} (\theta_i - \theta_{[m+1]})^{-2} + H^{\langle (\eta-1)/2 \rangle} (I_{\xi + \eta + 1}) . \label{eq:lb-H-1}
\end{multline}
By Lemma~\ref{lem:lb-instance-fact}, we have $\theta_{[m]}, \theta_{[m+1]} \in \frac{1}{2} +  Cn^{-1/4} \cdot (\frac{\xi}{8} \pm \frac{1}{100})$. Therefore, we have
\begin{multline}
    \left(\frac{n - \eta(2\eta + 1)}{2}  + \xi \eta\right) \cdot \left(\frac{1+C}{2} - \theta_{[m+1]}\right)^{-2} + \left(\frac{n - \eta(2\eta + 1)}{2}  - \xi \eta\right) \cdot \left(\frac{1-C}{2} - \theta_{[m]}\right)^{-2}\\ \in \left[\frac{C^{-2}n}{4}, 16C^{-2}n\right]. \label{eq:lb-H-2}
\end{multline}
and
\begin{multline}
     \sum_{j = 1}^{\xi + \eta} \sum_{\text{arm $i$ in block $j$}} (\theta_i - \theta_{[m]})^{-2} + \sum_{j = \xi+\eta+2}^{2\eta + 1} \sum_{\text{arm $i$ in block $j$}} (\theta_i - \theta_{[m+1]})^{-2} \\
     \leq \sum_{k=1}^{2\eta} \eta \cdot \left(C n^{-1/4} \cdot (k/16)\right)^{-2} \leq 256 \eta C^{-2} n^{1/2} \cdot \frac{\pi^2}{6} \leq 512 C^{-2} n^{3/4}. \label{eq:lb-H-3}
\end{multline}
Combining \eqref{eq:lb-H-1}, \eqref{eq:lb-H-2}, and \eqref{eq:lb-H-3}, for sufficiently large $n$, we have
\[
H^{\langle m \rangle}(I) \in \left[\frac{C^{-2}n}{4}, 17C^{-2}n\right] + H^{\langle (\eta-1)/2 \rangle} (I_{\xi + \eta + 1}) .
\]
Apply induction hypothesis to $H^{\langle (\eta-1)/2 \rangle} (I_{\xi + \eta + 1}) $ and we prove the lemma.
\end{proof}

\paragraph{Proof Intuition.} The intuition about our lower bound instance distribution is as follows. As pointed out in Lemma~\ref{lem:lb-instance-fact} (Item~\ref{lem:lb-instance-fact-item2}), the top $m = (n-1)/2$ arms consists of the top arms, the blocks from $I_{\xi+\eta+2}$ to $I_{2\eta + 1}$, and finally the top half (excluding the median) arms in block $I_{\xi + \eta+1}$. Therefore, two necessary tasks are i) to complete is to identify the value of $\xi$, and ii) to identify the top half arms in $I_{\xi + \eta+1}$. 

For the first task, in Section~\ref{sec:learn-bias}, we will introduce a sub-problem named ``learning the bias''. Via studying this problem, we will show that, any agent, if using at most $H^{\langle m \rangle}(I) / \sqrt{K}$ queries, cannot learn the correct value of $\xi$ with probability $\omega(n^{-1/4})$. Note that since there are only $2\eta + 1 = O(n^{1/4})$ possible values for $\xi$, this means that the agent cannot do much better than random guessing. Also, we note that we prove the impossibility statement for agents with even $\Theta(n C^{-2} / \ln(n/C))$ queries, which is a stronger statement as Lemma~\ref{lemma:lb-H} shows that $H^{\langle m \rangle}(I)$ is always $\tilde{\Theta}(nC^{-2})$. 

The above discussion suggests that a communication step is needed for the agents to collectively decide the exact value of $\xi$. It also suggests that not too many queries are made to the block $I_{\xi + \eta+1}$ before the first communication step (more specifically, the amount is at most $O(n^{-1/4}$ fraction of the total number of queries before the communication step, see Item~\ref{lemma:lb-first-round-item2} of Lemma~\ref{lemma:lb-first-round} for detailed justification), which is negligible for further identifying the top half arms in $I_{\xi + \eta+1}$ (the second necessary task). On the other hand, by Lemma~\ref{lemma:lb-H}, $H^{\langle (\eta -1)/2 \rangle}(I_{\xi + \eta + 1})$ is still $\tilde{\Theta}(nC^{-2})$. Therefore, we can recursively apply the similar argument to $I_{\xi + \eta + 1}$, yielding a communication round lower bound that is proportional to the number of hierarchies in the definition of $\mathcal{I}(C, \mu, n)$, which is $\Theta(\ln (\ln n/\ln k))$. 

This recursive (or inductive) argument is presented in Section~\ref{lb:final-recursion}. Note that in the simplified explanation above, we neglected the extra queries made to the $I_{\xi + \eta + 1}$ before the first communication step. To formally deal with these extra queries, we need to strengthen the inductive hypothesis, and introduce the definition of \emph{augmented algorithms} where the agents enjoy a small number of free (and shared) queries before the very first communication round. We will show that the communication round lower bound still holds even for augmented algorithms.

\subsubsection{The ``Learning the Bias'' Sub-problem and Its Analysis} \label{sec:learn-bias}

In this subsection, we identify a critical sub-problem for identifying the top $m$ arms in our constructed hard instances. We then prove the sample complexity lower bounds for the sub-problem, which will be a crucial building block for the ultimate lower bound theorem for identifying the top $m$ arms.

We first define the sub-problem as follows.

\medskip
\noindent {\bf Problem Definition} (Learning the Bias)\textbf{.} {\it
There are $n$ Bernoulli arms. Given the parameters $\epsilon \in (0, 1/8)$, $\mu \in (3/8, 5/8)$, and a distribution $\D$ supported on $\{\pm 1\}^n$ (which are publicly known), a hidden vector $(b_1, b_2, \dots, b_n)$ is sampled from $\D$, and the mean reward of the $i$-th arm is set to be $\theta_i = \mu + b_i \epsilon$ (also hidden from the algorithm). The algorithm has a budget of $T$ adaptive samples from the arms, and the goal is to decide the bias $B = b_1 + b_2 + \dots + b_n$.
}

The following lemma shows the sample complexity lower bound for learning the bias when $\D$ is the uniform distribution.

\begin{lemma}\label{lemma:lb-bias}
Assume that $\D$ is the uniform distribution of $\{\pm 1\}^n$. For sufficiently large $n$, if \[T \leq n \epsilon^{-2} / (20000\ln (n/\epsilon))\,,\] then the probability that the player correctly identifies $B$ is at most $O(n^{-1/2})$.
\end{lemma}

\begin{proof}
Without loss of generality we can assume that the player makes the guess about $B$ after using all $T$ samples. Let $\tau = (i_1, y_1, i_2, y_2, \dots, i_T, y_T)$ be the \emph{transcript} of all samples, where $i_t \in [n]$ denotes the arm sampled from at time $t$, and $y_t \in \{0, 1\}$ denotes the observation at time $t$. The player will finally uses an algorithm $\A(\tau)$ to decide the guess about the bias. The probability that the player makes a correct guess is 
\begin{multline}
    \Pr[\text{correct guess}] = \mathop{\bE}_{\tau} \Pr_{(b_1, \dots b_n) \sim \D} \left[\A(\tau) = B | \tau \right]  \\ \leq \mathop{\bE}_\tau \max_{\beta} \left\{\Pr_{(b_1, \dots b_n) \sim \D} \left[B = \beta | \tau \right]\right\} = \mathop{\bE}_\tau \max_{\beta} \left\{\Pr_{(b_1, \dots b_n) \sim \D(\tau)} \left[B = \beta\right]\right\} , \label{eq:lb-bias-00}
\end{multline}
where we let $\D(\tau)$ be the posterior distribution of $(b_1, b_2, \dots, b_n)$ given $\tau$.

For fixed $\tau$, let $r_{i, t}$ be the number of $1$'s the player observes among the first $t$ samples made from the $i$-th arm. Also let $T_i$ be the total number of samples made from the $i$-th arm, we have $T = T_1 + T_2 + \dots + T_n$.  Since $\D$ is a product distribution, we have
\begin{align}
   \D(\tau) = \otimes_{i=1}^{n} \D_i(r_{i, T_i}, T_i), \label{eq:lb-bias-0}
\end{align}
where $\D_i(r_{i, T_i}, T_i)$ is the posterior of $b_i$ given that $r_{i, T_i}$ 1's are observed from $T_i$ samples from arm $i$.

 Note that the posterior distribution  $\D_i(r_{i, T_i}, T_i)$ is completely determined by the posterior probability $p_i \triangleq
\Pr\left[b_i = +1 | r_{i, T_i}, T_i\right]$.  Let $\sigma_i^2 \triangleq \mathrm{Var}[\D_i(r_{i, T_i}, T_i)] = 4p_i(1-p_i)$ for every arm $i$. Let $\sigma^2 \triangleq \sigma_1^2 + \sigma_2^2 + \dots + \sigma_n^2$ (where $\sigma \geq 0$). By \eqref{eq:lb-bias-0} and invoking the Berry-Esseen theorem (Theorem~\ref{thm:berry-esseen}) with $\rho = 8$, we have
\[
\forall x \in \mathbb{R}, \Pr_{(b_1, \dots b_n) \sim \D(\tau)}\left[\sigma^{-1} \sum_{i=1}^{n} (b_i - (1-2p_i)) \leq x\right] \in  \Phi(x) \pm 4.5 n \sigma^{-3/2}.
\]
Since $\Phi(x)$ is a continuous function, we have
\begin{align}
    \forall \beta \in \mathbb{R}, \Pr_{(b_1, \dots b_n) \sim \D(\tau)} [B = \beta] \leq 9n \sigma^{-3/2}. \label{eq:lb-bias-0a}
\end{align}

Now we will estimate the posterior probability $p_i$ and give a lower bound on $\sigma^2$ to upper bound the probability in \eqref{eq:lb-bias-0a}.

Via standard concentration inequalities (e.g., Hoeffding's inequality), we have
\[
\forall i, t, \Pr_\tau\left[|r_{i, t} - t \theta_i | \leq 2 \sqrt{t \ln (n/\epsilon)}\right] \geq 1 - 2 (\epsilon/n)^8 .
\]
Therefore, if we define $\E_4$ be the event 
\[
\E_4 \triangleq \{\forall i, t, |r_{i, t} - t \theta_i | \leq 2 \sqrt{t \ln (n/\epsilon)}\},
\]
we have $\Pr_\tau[\E_4] \geq 1 - 1/n^6$.

Say an arm $i$ is \emph{sufficiently explored} if $T_i \geq \epsilon^{-2} / (10000 \ln(n/\epsilon))$. By Markov's inequality, there are at most $n/2$ sufficiently explored arms. Fix an arm $i$, let $r = r_{i, T_i}$ for notational convenience.  Conditioned on the event $\E_4$ and that it is insufficiently explored, we have 
$|r - T_i \theta_i | \leq 1/(50\epsilon).$ Since $\theta_i = \mu \pm \epsilon$, we further have 
\begin{align}
\left|r - \mu T_i\right| \leq \frac{1}{50\epsilon} + \epsilon T_i \leq \frac{1}{40\epsilon} . \label{eq:lb-bias-1}
\end{align}
By the definition of $p_i$, we have
\begin{align}
p_i = \frac{(\mu + \epsilon)^r (1 - \mu - \epsilon)^{T_i - r}}{(\mu + \epsilon)^r (1- \mu - \epsilon)^{T_i - r} + (\mu - \epsilon)^r (1 - \mu + \epsilon)^{T_i - r}}
=  \frac{1}{1 + (\frac{\mu -\epsilon}{\mu + \epsilon})^{r} (\frac{1 - \mu + \epsilon}{1 - \mu - \epsilon})^{T_i - r}} . \label{eq:lb-bias-2}
\end{align}
By \eqref{eq:lb-bias-1}, for $\epsilon \in (0, 1/8)$ and $\mu \in (3/8, 5/8)$, we have that

Note that
\begin{align}
\left(\frac{\mu -\epsilon}{\mu + \epsilon}\right)^{r} \left(\frac{1 - \mu + \epsilon}{1 - \mu - \epsilon}\right)^{T_i - r} = \left(\frac{\mu -\epsilon}{\mu + \epsilon}\right)^{\mu T_i + (r - \mu T_i)} \left(\frac{1 - \mu + \epsilon}{1 - \mu - \epsilon}\right)^{(1-\mu)T_i + (\mu T_i - r)} ,  \label{eq:lb-bias-3}
\end{align}
and for $\epsilon \in (0, 1/8)$, $\mu \in (3/8, 5/8)$, and $T_i \leq \epsilon^{-2}/10000$, it holds that
\begin{align}
   \left(\frac{\mu -\epsilon}{\mu + \epsilon}\right)^{\mu T_i} \left(\frac{1 - \mu + \epsilon}{1 - \mu - \epsilon}\right)^{(1-\mu)T_i} \in [0.99, 1.01].  \label{eq:lb-bias-3a}
\end{align}
Also, by \eqref{eq:lb-bias-1} and the ranges for $\epsilon$ and $\mu$, we have
\begin{multline}
     \left(\frac{\mu -\epsilon}{\mu + \epsilon}\right)^{r - \mu T_i} \left(\frac{1 - \mu + \epsilon}{1 - \mu - \epsilon}\right)^{\mu T_i - r}\\  \in \left[\left(\frac{(\mu-\epsilon)(1-\mu-\epsilon)}{(\mu+\epsilon)(1 -\mu+\epsilon)}\right)^{1/(40\epsilon)}, \left(\frac{(\mu+\epsilon)(1-\mu+\epsilon)}{(\mu-\epsilon)(1 -\mu-\epsilon)}\right)^{1/(40\epsilon)}\right] \subseteq [0.75, 1.3] . \label{eq:lb-bias-3b}
\end{multline}
Combining \eqref{eq:lb-bias-2}, \eqref{eq:lb-bias-3}, \eqref{eq:lb-bias-3a}, and \eqref{eq:lb-bias-3b}, and conditioned on $\E_4$ and that arm $i$ is insufficiently explored, we have that $p_i \in [0.4, 0.6]$, meaning that $\sigma_i^2 \geq 0.96$. Conditioned on $\E_4$, since there are at most $n/2$ sufficiently explored arms, we have that $\sigma^2 \geq 0.48n$. Using \eqref{eq:lb-bias-0a}, we have
\begin{multline*}
\mathop{\bE}_\tau \max_{\beta} \left\{\Pr_{(b_1, \dots b_n) \sim \D(\tau)} \left[B = \beta\right]\right\} \leq \mathop{\bE}_\tau\left[ \max_{\beta} \left\{\Pr_{(b_1, \dots b_n) \sim \D(\tau)} \left[B = \beta\right]\right\} \Big| \E_4\right] + \Pr[\overline{\E_4}]\\
\leq 9n (0.48n)^{-3/2} + n^{-6} \leq 30 n^{-1/2}.
\end{multline*}
Together with \eqref{eq:lb-bias-00}, we prove the lemma.
\end{proof}

To analyze the lower bound for our hard instances $\I(C, \mu, n)$, we need to adapt Lemma~\ref{lemma:lb-bias} to a different distribution $\D$, as shown in the following corollary.
\begin{corollary}\label{cor:lb-bias}
For any $S \subseteq [n]$ ($S \neq \emptyset$), let $\D = \D(S)$ be the following distribution supported on $\{\pm 1\}^n$: first sample a uniformly random integer $s$ from $S$, and then sample a uniformly random vector from $\{\pm 1\}^n$ such that the number of $+1$'s in the vector is exactly $s$. Let $q = \min_{s \in S}\{2^{-n} {\binom{n}{s}}\}$. For sufficiently large $n$, if $T \leq n \epsilon^{-2}/(20000 \ln (n/\epsilon))$, then in the learning the bias problem, the probability that the player correctly identifies $B$ is at most $O(n^{-1/2}q^{-1}\cdot |S|^{-1})$
\end{corollary}
\begin{proof}
Construct the joint distribution with probability mass function $p$ for the random variables \[(b_1, b_2, \dots, b_n, Y) \in \{\pm 1\}^{n} \times \{0, 1\}\] as follows (where we let $\beta = \sum_{i=1}^{n} (b_i+1)/2$ be the number of $+1$'s in the vector $(b_1, b_2, \dots, b_n)$),
\[
p(b_1, b_2, \dots, b_n, Y) = \left\{
\begin{array}{ll}
0 & \text{when $\beta \not\in S$ and $Y=1$}\\
2^{-n} {\binom{n}{\beta}} & \text{when $\beta \not\in S$ and $Y=0$}\\
q{{\binom{n}{\beta}}}^{-1} & \text{when $\beta \in S$ and $Y = 1$}\\
2^{-n} {\binom{n}{\beta}} - q{{\binom{n}{\beta}}}^{-1} & \text{when $\beta \in S$ and $Y = 0$}
\end{array}
\right. .
\]
It is clear that the marginal distribution on $(b_1, b_2, \dots, b_n)$ is uniform over $\{\pm 1\}^{n}$ and the conditional distribution on $(b_1, b_2, \dots, b_n)$ given that $Y = 1$ is $\D(S)$. Therefore, the probability that the player correctly guesses $B$ given $\D = \D(S)$ is 
\begin{multline*}
    \Pr_{(b_1, b_2, \dots, b_n) \sim \D(S)}[\text{correct guess}] = \Pr_{(b_1, b_2, \dots, b_n, Y) \sim p}[\text{correct guess} | Y = 1] \\
    \leq \frac{\Pr_{(b_1, b_2, \dots, b_n, Y) \sim p}[\text{correct guess}]}{\Pr_{(b_1, b_2, \dots, b_n, Y) \sim p}[Y = 1]} \leq O(n^{-1/2}) \cdot \frac{1}{q |S|},
\end{multline*}
where in the last inequality, we invoked Lemma~\ref{lemma:lb-bias}.
\end{proof}

\subsubsection{The Lower Bound Theorems for Communication Rounds and Concluding the Proof}\label{lb:final-recursion}

The following lemma helps to relate the lower bound results derived for the learning the bias problem in Section~\ref{sec:learn-bias} to the form of top $m$ arm identification.

\begin{lemma}\label{lemma:lb-first-round}
For sufficiently large $K$, any odd integer $n$ such that $n > K^{10}$, any $C \in (0, 1/8)$, and any $\mu \in (3/8, 5/8)$, consider a random instance from $\I(C, \mu, n)$. For any player that  makes at most $T$ sequential samples, when $T \leq nC^{-2} / (40000\ln (n/C))$, we have the following claims.
\begin{enumerate}
    \item The probability that the player correctly identifies the top $m$ arms is at most $O(n^{-1/4})$ (recall that $m = (n-1)/2$).
    \item The expected number of samples made to any arm in the block that contains the median arm  is at most $O(n^{-1/4} T)$. \label{lemma:lb-first-round-item2}
\end{enumerate}
\end{lemma}

\begin{proof}
We prove the lemma by reducing the learning of the bias problem to the top-$m$-arm identification problem for the instance distribution $\I(C, \mu, n)$

Let us consider the learning the bias problem with $n' = n - \eta(2\eta + 1)$ arms (note that $n' \geq n/2$), $\epsilon = C$, the same $\mu$ parameter, and the distribution $\D = \D(S)$ where $S = \{n'/2 + z \eta : z \in \{-\eta, \eta + 1, \dots, \eta -1, \eta\}\}$. Once the expected rewards of the $n'$ arms are determined by a sample from $\D$, let $J$ be set of the arms. To construct a top-$m$-arm identification problem instance, we independently sample smaller problem instances $I_1, I_2, \dots, I_{2\eta + 1}$ such that $I_j \sim \I(C\sqrt{\eta/n}, \mu + \frac{j - \eta - 1}{8} \cdot C n^{-1/4}, \eta)$ for each $j \in [2\eta + 1]$. Let $I = J \cup I_1 \cup I_2 \cup \dots \cup I_{2\eta + 1}$. One can verify that $I$ follows the distribution $\I(C, \mu, n)$.

Now we prove the first claim. Suppose that the player correctly identifies the top $m$ arms with probability $p$. According to Lemma~\ref{lem:lb-instance-fact}, only the arms in block $(z + \eta +1)$ have non-empty intersection with both the set of top $m$ arms and the set of the remaining arms. Therefore, when conditioned on that the top $m$ arms are correctly identified, by checking the identified set of arms, the player can find out the value of $z$, and deduce that the bias $B = 2z \eta$. Therefore, there exists an algorithm correctly identifying the bias with probability at least $p$. Invoking Corollary~\ref{cor:lb-bias}, and noting that $|S| = \Theta(n^{-1/4})$, $q = \Theta(1)$, we have that 
\[
p \leq O\left((n')^{-1/2}q^{-1} \cdot |S|^{-1}\right) \leq O(n^{-1/4}).
\]

Regarding the second claim, let us consider the player, after making at most $T$ samples, guessing $z = \tilde{z}$ with probability $t_{\tilde{z} + \eta + 1}/T$, where $t_{j}$ is the number of samples made to the $j$-th block, and finally guessing $B = 2\tilde{z}{\eta}$. The probability that the player successfully identifies the bias $B$ is $\bE[t_{z + \eta + 1}/T]$. Invoking Corollary~\ref{cor:lb-bias}, we have that this value is upper bounded by $O(n^{-1/4})$. Therefore, we have $\bE[t_{z + \eta + 1}] \leq O(n^{-1/4}T)$, which proves the claim by noting that block $(z+\eta+1)$ contains the median arm in $I$, due to Lemma~\ref{lem:lb-instance-fact}.
\end{proof}

\paragraph{Definition of Augmented Algorithms and Uniform Upper Bound $\mathfrak{p}$.} For fixed and sufficiently large $K$, let $\mathcal{A}^{(\alpha)}_{R, T}$ be the set of $R$-round $K$-agent \emph{augmented} algorithms defined as follows. Any algorithm in $\mathcal{A}^{(\alpha)}_{R, T}$ has $R$ rounds of communication, where during each round, the time budget is $T/\sqrt{K}$. Before the first round, there is an \emph{augmented round} where a single thread is allowed to make $\alpha T$ sequential samples and broadcast the observations to all agents. Let
\begin{align}
    \mathfrak{p}^{(\alpha)}_{R, T}(C, n) \triangleq \sup_{\mathbb{A} \in \mathcal{A}^{(\alpha)}_{R, T}} \sup_{\mu: |\mu - 1/2| < 1/8 - C} \Pr_{I \sim \I(C, \mu, n), \mathbb{A}} [\text{$\mathbb{A}$ identifies the top $m = (n-1)/2$ arms}] \label{eq:lb-frakp-def}
\end{align}
be the best success probability of augmented algorithms in $\mathcal{A}^{(\alpha)}_{R, T}$ when the input instance follows $\I(C, \mu, n)$ for any $\mu : |\mu - 1/2| < 1/8 - C$, where the subscript of $\Pr$ specifies that the probability is taken over both $I$ and the randomness of algorithm $\mathbb{A}$. Clearly, if we can prove that $\mathfrak{p}^{(\alpha)}_{R, T}(C, n) < 0.1$ for any $\alpha$, we obtain the round complexity lower bound $R$ for fixed-time algorithms with time budget $T$. In what follows, we will prove upper bounds for  $\mathfrak{p}^{(\alpha)}_{R, T}(C, n)$ via induction.

\begin{lemma}\label{lemma:lb-main-recurrence}
For any positive $C \in (0, 1/8)$, any constant $\iota \geq 0$, suppose $n > K^{10}$, $R \geq 2$, and $\sqrt{K} \geq 80000 \ln^{\iota+1} (nC^{-2})$. Let $T = (n C^{-2}) \ln^\iota (n C^{-2})$, and let $\eta$ be the smallest odd integer that is greater than $\sqrt[4]{n}$. It holds that
\[
\mathfrak{p}^{(n^{-1/8})}_{R, T}(C, n) \leq O\left(n^{-3/32}\right) + \mathfrak{p}^{(\eta^{-1/8})}_{R-1, T}(C\sqrt{\eta/n}, \eta) .
\]
\end{lemma}
\begin{proof}
Fix any algorithm $\mathbb{A} \in \mathcal{A}^{(\alpha)}_{R, T}$ and any $\mu \in (3/8 + C, 5/8-C)$, we will upper bound the success probability of $\mathbb{A}$ given the input instance $I \sim \I(C, \mu, n)$. 


Recall in the construction of the instance $I \sim \I(C, \mu, n)$, $(2\eta+1)$ blocks are independently sampled. Let $\zeta \in [2\eta+1]$ be the block which the median arm is in. Let $n_0$ be the number of samples made in the augmented round to arms in block $\zeta$, and let $n_i$ be the number of samples made in the first round to arms in block $\zeta$ by agent $i$. For every agent $i$, by the second claim of Lemma~\ref{lemma:lb-first-round} (and  observing that the total number of samples made in the augmented round and the first round by agent $i$ is at most $T /\sqrt{K} + n^{-1/8} T \leq T (K^{-1/2} + K^{-{5/4}}) \leq nC^{-2}/(40000 \ln (n/C))$), we have that 
\[
\mathop{\bE}_{I \sim \I(C, \mu, n), \mathbb{A}}[n_0 + n_i]\leq O \left(nC^{-2} \cdot n^{-1/4}\right) .
\]
Therefore, 
\[
\mathop{\bE}_{I \sim \I(C, \mu, n), \mathbb{A}}[n_0 + n_1 + n_2 + \dots + n_K] \leq O\left(n^{3/4} C^{-2} \cdot K\right) \leq O\left(n^{7/8} C^{-2}\right).
\]
Let $\E_5$ be the event that $n_0 + n_1 + n_2 + \dots + n_K \leq n C^{-2} \cdot \eta^{-1/8}$. By Markov's inequality, we have
\begin{align}
\Pr_{I \sim \I(C, \mu, n), \mathbb{A}} [\E_5] \geq 1 -\frac{O(n^{7/8}C^{-2})}{nC^{-2}\cdot \eta^{-1/8}} = 1-O\left(n^{-3/32}\right). \label{eq:lb-main-recurrence-3}
\end{align}

Our next goal is to establish \eqref{eq:lb-main-recurrence-4} for every $j \in [2\eta + 1]$. Fix such $j$, consider the following algorithm $\mathbb{B}$ that works for an instance $I_j \sim \I(C\sqrt{\eta/n}, \mu + \frac{j-\eta-1}{8}\cdot Cn^{-1/4}, \eta)$. The algorithm first samples $I_{j'} \sim \I(C\sqrt{\eta/n}, \mu + \frac{j'-\eta-1}{8}\cdot Cn^{-1/4}, \eta)$ for all $j' \neq j$. The algorithm $\mathbb{B}$ also creates $((n - \eta(2\eta+1))/2 + (j-\eta-1)\eta)$ Bernoulli arms with mean reward $(\mu + C/2)$, and $((n - \eta(2\eta+1))/2 - (j-\eta-1)\eta)$ Bernoulli arms with mean reward $(\mu - C/2)$. Combining all the arms (including those in $I_1, I_2, \dots, I_{2\eta+1}$), we have an instance $I^{\flat}$ of $n$ arms.  The algorithm $\mathbb{B}$ simulates algorithm $\mathbb{A}$ with input instance $I^{\flat}$ in the following manner. Whenever $\mathbb{A}$ is to sample an arm in $I_j$, $\mathbb{B}$ queries the real arm in $I_j$, otherwise $\mathbb{B}$ simulates a sample to the artificial arm, and feed the observation to $\mathbb{A}$. More importantly, only a single thread is used to simulate $\mathbb{A}$ during the augmented round and the first round. Then, if event $\E_5$ holds, all $K$ agents are used to simulate the corresponding agents in $\mathbb{A}$ from the second round and reports $I_j$ intersecting the set of top arms returned by $\mathbb{A}$; otherwise, $\mathbb{B}$ reports failure and terminates.

Note that $\mathbb{B} \in \mathcal{A}_{R-1,T}^{(\eta^{-1/8})}$. By the definition in \eqref{eq:lb-frakp-def}, we have
\begin{align*}
\Pr_{I_j, \mathbb{B}} [\text{$\mathbb{B}$ identifies the top $(\eta-1)/2$ arms in $I_j$}] \leq \mathfrak{p}^{(\eta^{-1/8})}_{R-1, T}(C\sqrt{\eta/n}, \eta) . 
\end{align*}
Note that $I^{\flat}$ constructed above follows the conditional distribution $\I(C, \mu, n)$ given that the median arm is in the $j$-th block. Also note that when $\E_5$ holds and  $\mathbb{A}$ is correct, $\mathbb{B}$ is also correct. Therefore, we have
\begin{align}
&\Pr_{I\sim\I(C,\mu, n), \mathbb{A}} [\text{$\mathbb{A}$ identifies the top $(n-1)/2$ arms in $I$} | \zeta = j] \nonumber \\
\leq& \Pr_{I\sim\I(C,\mu, n), \mathbb{A}} [\text{$\mathbb{A}$ identifies the top $(n-1)/2$ arms in $I$} \wedge \E_5 | \zeta = j] + \Pr_{I\sim\I(C,\mu, n), \mathbb{A}} [\overline{\E_5} | \zeta = j] \nonumber  \\
\leq& \Pr_{I_j, \mathbb{B}} [\text{$\mathbb{B}$ identifies the top $(\eta-1)/2$ arms in $I_j$}] + \Pr_{I\sim\I(C,\mu, n), \mathbb{A}} [\overline{\E_5} | \zeta = j] \nonumber  \\
\leq& \mathfrak{p}^{(\eta^{-1/8})}_{R-1, T}(C\sqrt{\eta/n}, \eta) + \Pr_{I\sim\I(C,\mu, n), \mathbb{A}} [\overline{\E_5} | \zeta = j] , \label{eq:lb-main-recurrence-4}
\end{align}
where $\zeta$ is the block which the median arm is in.

Since \eqref{eq:lb-main-recurrence-4} holds for all $j \in [2\eta + 1]$. We have
\begin{align}
\Pr_{I\sim\I(C,\mu, n), \mathbb{A}} [\text{$\mathbb{A}$ identifies the top $(n-1)/2$ arms in $I$}] \leq \mathfrak{p}^{(\eta^{-1/8})}_{R-1, T}(C\sqrt{\eta/n}, \eta) + \Pr_{I\sim\I(C,\mu, n), \mathbb{A}} [\overline{\E_5}] . \label{eq:lb-main-recurrence-5}
\end{align}
Combining \eqref{eq:lb-main-recurrence-3} and \eqref{eq:lb-main-recurrence-5}, we conclude the proof of the lemma.
\end{proof}

When $R=1$, the agents do not communicate except for the shared observation from the augmented round. Therefore, Lemma~\ref{lemma:lb-first-round} implies that for $C \in (0, 1/8)$,  any constant $\iota \geq 0$, $n > K^{10}$, $\sqrt{K} \geq 80000 \ln^{\iota + 1}(nC^{-2})$, and $T = (nC^{-2})\ln^{\iota} (nC^{-2})$, it holds that
\begin{align}
\mathfrak{p}^{(n^{-1/8})}_{1, T}(C, n) \leq O\left(n^{-1/4}\right). \label{eq:lb-main-recurrence-base}
\end{align}

Combining Lemma~\ref{lemma:lb-main-recurrence} and \eqref{eq:lb-main-recurrence-base}, we have the following lemma.
\begin{lemma}\label{lemma:lb-loglogm}
For $C \in (0, 1/8)$,  any constant $\iota \geq 0$, $n > K^{10}$, $\sqrt{K} \geq 80000 \ln^{\iota + 1}(nC^{-2})$, $T = (nC^{-2})\ln^{\iota} (nC^{-2})$, and $R = \lceil \log_{4} \frac{\ln n}{10 \ln K}\rceil$, it holds that
\[
\mathfrak{p}^{(n^{-1/8})}_{R, T}(C, n) \leq \sum_{r=0}^{R-1} O\left(n^{-4^{-r} \cdot 3/32}\right) \leq O(K^{-1/5}).
\]
\end{lemma}

Theorem~\ref{thm:lb-loglogm} is proved by setting $\iota = 1$ and $C = 1/10$ in Lemma~\ref{lemma:lb-loglogm}, together with Lemma~\ref{lemma:lb-H} and the discussion below the definition of $\mathfrak{p}$ (Equation \eqref{eq:lb-frakp-def}).

\section{The Fixed-Confidence Case}
\label{sec:confidence}

In this section we discuss the fixed-confidence case.  We first present a collaborative algorithm for the fixed-confidence case. The algorithm is inspired by~\cite{HKKLS13} and~\cite{BWV13}, and described in Algorithm~\ref{alg:confidence}.

\begin{algorithm}[t]
\caption{Collaborative algorithm for fixed-confidence setting.}
\label{alg:confidence}
\KwIn{a set of arms $I$, parameter $m$, and a confidence parameter $\delta$.}
\KwOut{a set of top-$m$ arms of $I$.}
Initialize $I_0 \gets I$, $m_0 \gets m$, $
\Acc_0 \gets \emptyset$, $\Rej_0 \gets \emptyset$, $r \gets 0$, $T_{-1} \gets 0$\;
for $r = 0, 1, \dotsc$, let $\eps_r = 2^{-(r+1)}$ and $T_r = {8 \log (4 n (r+1)^2 \delta^{-1})}/(K \eps^2_{r})$\;
\While{$I_{r} \neq \emptyset$}{
each agent pulls each arm in $I_{r}$ for $T_{r} - T_{r - 1}$ times\;
for each $i \in I_{r}$, let $\hat{\theta}^{(r)}_i$ be the estimated mean of the $i$-th arm in $I_r$ after $K T_r$ pulls (over all rounds and agents so far)\;
let $\pi_r : \{1, \dotsc, \abs{I_{r}}\} \to I_{r}$ be the bijection such that $\hat{\theta}^{(r)}_{\pi_r(1)} \ge \hat{\theta}^{(r)}_{\pi_r(2)} \ge \dotsc \ge \hat{\theta}^{(r)}_{\pi_r(|I_r|)}$\;
$\Acc_{r + 1} \gets \Acc_r \cup \{i \in I_{r} : \hat{\theta}^{(r)}_i > \hat{\theta}^{(r)}_{\pi_r(m_r + 1)} + \eps_r\}$\;
$\Rej_{r+1} \gets \Rej_r \cup \{i \in I_{r} : \hat{\theta}^{(r)}_i < \hat{\theta}^{(r)}_{\pi_r(m_r)} - \eps_r\}$\;
$m_{r + 1} \gets m - \abs{\Acc_{r+1}} $\;
$I_{r + 1} \gets I_{r} \setminus \left(\Acc_{r+1} \cup \Rej_{r + 1}\right)$\;
$r \gets r + 1$\;
}
\KwRet{$\Acc_r$.}
\end{algorithm}

\begin{theorem}
\label{thm:confidence-up}
There is an algorithm (Algorithm~\ref{alg:confidence}) that solves top-$m$ arm identification with probability at least $1-\delta$, using $O\left(\log \left(1/\Delta^{\langle m\rangle}_{[m]}\right)\right)$ rounds of communication and 
$O\left(\frac{H^{\langle m\rangle}}{K} \log\left(\frac{n}{\delta} \log{H^{\langle m\rangle}}\right)\right)$ time.
\end{theorem}

\begin{proof}
First, by Chernoff-Hoeffding we have that for any $r \ge 0$ and $i \in I_r$,
\[\Pr\left[\abs{\hat{\theta}^{(r)}_i - \theta_i} \ge \frac{\eps_r}{4}\right]\le 2 \exp\left(-\frac{\eps^2_r}{8} K T_r \right) \le \frac{\delta}{2n(r + 1)^2}.\] 
By a union bound, the event $\E_3 : \forall{i, r}, \abs{\hat{\theta}^{(r)}_i - \theta_i} \le {\eps_r/4}$ holds with probability at least $1 - \delta$.

It suffices to show that conditioned on event $\E_3$, the algorithm does not make any error and terminates using the stated time and rounds. We  prove this by induction on $r$.  We have the following induction hypothesis:
\begin{enumerate}
    \item \label{ind:i1} $m_r = |\Top_m \cap I_r|$,
    \item \label{ind:i2}$\Acc_r \subseteq \Top_m$ and $\Rej_r \subseteq I \setminus \Top_m$,
    \item \label{ind:i3} $\{i \in I \mid \Delta^{\langle m\rangle}_i \ge 4 \eps_r \} \cap I_{r+1} =\emptyset$.
\end{enumerate}

It is easy to see that the base case ($r = 0$) holds trivially. Let us assume that the hypothesis holds for round $(r - 1)$, and consider round $r$.  By event $\E_3$ we have 
\begin{equation}
\label{eq:x-1}
\forall{i \in I_r : \abs{\hat{\theta}^{(r)}_i - \theta_i} \le \eps_r/4}.
\end{equation}
And for any $a \in \{1, \dotsc, |I_{r}|\}$ we have
\begin{equation}
\label{eq:x-2}
\theta_{[a]}(I_r) - \eps_r / 4 \le \hat{\theta}^{(r)}_{\pi_r(a)} \le \theta^{(r)}_{[a]}(I_r) + \eps_r/4\,.    
\end{equation}
If $\hat{\theta}^{(r)}_{i} > \hat{\theta}^{(r)}_{\pi_r(m_r + 1)} + \eps_r$, in which case the algorithm adds $i$ to $\Acc_{r+1}$, then by (\ref{eq:x-1}) and (\ref{eq:x-2}) we have
\begin{equation*}
\theta_i + \eps_r / 4 \ge \hat{\theta}^{(r)}_{i} > \hat{\theta}^{(r)}_{\pi_r(m_r + 1)} + \eps_r \ge \theta_{[m_r + 1]}(I_r) +  3\eps_r/4.
\end{equation*}
We thus have $\theta_i - \theta_{[m_r+1]}(I_r) > \eps_r/2$, which implies that $\theta_i \ge\theta_{[m_r]}(I_r)$. By the first and second items of the induction hypothesis, we have $i \in \Top_m$, which implies $m_{r+1} = \abs{\Top_m \cap I_{r+1}}$ and $\Acc_{r + 1} \subseteq \Top_m$. Similarly we can also show $\Rej_{r + 1} \subseteq I \setminus \Top_m$.

We next consider the third item of the induction hypothesis.   For an arm $i \in I_r \supseteq I_{r+1}$ such that  $\Delta_i^{\langle m \rangle} \ge 4\eps_r$ and $\theta_i \ge \theta_{[m]}$, we have
\begin{eqnarray}
\label{eq:x-3}
 \hat{\theta}^{(r)}_i &\ge& \theta_i - {\eps_r}/{4} = \Delta^{\langle m \rangle}_i + \theta_{[m + 1]} - {\eps_r}/{4} \nonumber \\
 &\ge& \theta_{[m_r + 1]}(I_r) + \Delta^{\langle m\rangle}_i - {\eps_r}/{4} \nonumber \\
 &\ge& \hat{\theta}^{(r)}_{\pi_r(m_r  + 1)} + \Delta^{\langle m\rangle}_i - {\eps_r}/{2} \nonumber \\
 &\ge& \hat{\theta}^{(r)}_{\pi_r(m_r  + 1)} + 2\eps_r. \nonumber
\end{eqnarray}
Thus the $i$-th item in $I$ will be added into $\Acc_{r + 1}$, and thus will not appear in $I_{r+1}$.  By the same line of arguments, we can show that for any arm $i \in I_r$, if $\Delta_i^{\langle m \rangle} \ge 4\eps_r$ and $\theta_i \le \theta_{[m+1]}$, then $i \in \Rej_{r+1}$ and will not appear in $I_{r+1}$.  

With the three items in the induction hypothesis, we prove the correctness of the algorithm and analyze its time and round complexities.
By the definition of $\eps_r$, when $r \ge r_0 =  \log (4/\Delta^{\langle m\rangle}_i)$, we have $\{i \in I \mid \Delta^{\langle m\rangle}_i \ge 4 \eps_r\} = I$.  We have the followings: 
\begin{enumerate}
\item By the third item of the induction hypothesis, the algorithm will terminate in $r_0$ rounds.  

\item By the second item of the induction hypothesis, we have $\Acc_{r_0} = \Top_m$. 

\item Note that if $\eps_r \le \Delta^{\langle m \rangle}_i / 4$, then each agent pulls the $i$-th arm for at most $T_r = 8 \log(4 n (r+1)^2 \delta^{-1}) / (K \eps^2_{r})$ times.  Let $r(i) = \min_r \{\eps_r \le \Delta^{\langle m\rangle}_i / 4\}$; we thus have $\Delta^{\langle m\rangle}_i / 8 \le \eps_{r(i)} \le \Delta^{\langle m\rangle}_i / 4$. By the third item of the induction hypothesis, each agent pulls the $i$-th arm for at most 
$$
T_{r(i)} \le \frac{512}{K \left(\Delta^{\langle m\rangle}_i\right)^2} \log\left(\frac{16 n}{\delta} \log^2\left(4/\Delta^{\langle m\rangle}_i\right)\right)
$$
times.  Therefore, the total running time is bounded by 
$
\sum\limits_{i \in I} T_{r(i)} = O\left(\frac{H^{\langle m\rangle}}{K} \log\left({\frac{n}{\delta} \log H^{\langle m\rangle}}\right)\right).
$
\end{enumerate}
\end{proof}

Finally we comment on the lower bound. In \cite{TZZ19} it was shown that for the special case when $m = 1$, to achieve a running time of $\tilde{O}(H^{\langle 1 \rangle}/K)$ with success probability $0.99$ one needs at least $\log \left(1/\Delta_{[1]}^{\langle 1 \rangle}\right)$ rounds.  Therefore the upper bound in Theorem~\ref{thm:confidence-up} is  tight up to logarithmic factors.

\bibliographystyle{plain}
\bibliography{paper}

\begin{thebibliography}{10}

\bibitem{AAKS13}
Ittai Abraham, Omar Alonso, Vasilis Kandylas, and Aleksandrs Slivkins.
\newblock Adaptive crowdsourcing algorithms for the bandit survey problem.
\newblock In {\em {COLT}}, pages 882--910, 2013.

\bibitem{AAAK17}
Arpit Agarwal, Shivani Agarwal, Sepehr Assadi, and Sanjeev Khanna.
\newblock Learning with limited rounds of adaptivity: Coin tossing, multi-armed
  bandits, and ranking from pairwise comparisons.
\newblock In {\em COLT}, pages 39--75, 2017.

\bibitem{AS15}
Yossi Arjevani and Ohad Shamir.
\newblock Communication complexity of distributed convex learning and
  optimization.
\newblock In {\em NIPS}, pages 1756--1764, 2015.

\bibitem{AKZ19}
Sepehr Assadi, Nikolai Karpov, and Qin Zhang.
\newblock Distributed and streaming linear programming in low dimensions.
\newblock In {\em PODS}, pages 236--253. ACM, 2019.

\bibitem{ABM10}
Jean{-}Yves Audibert, S{\'{e}}bastien Bubeck, and R{\'{e}}mi Munos.
\newblock Best arm identification in multi-armed bandits.
\newblock In {\em COLT}, pages 41--53, 2010.

\bibitem{AK05}
Baruch Awerbuch and Robert~D. Kleinberg.
\newblock Competitive collaborative learning.
\newblock In {\em COLT}, pages 233--248, 2005.

\bibitem{BXJW19}
Yu~Bai, Tengyang Xie, Nan Jiang, and Yu-Xiang Wang.
\newblock Provably efficient q-learning with low switching cost.
\newblock In {\em NeurIPS}, 2019.

\bibitem{BBFM12}
Maria{-}Florina Balcan, Avrim Blum, Shai Fine, and Yishay Mansour.
\newblock Distributed learning, communication complexity and privacy.
\newblock In {\em {COLT}}, pages 26.1--26.22, 2012.

\bibitem{BL18}
Ilai Bistritz and Amir Leshem.
\newblock Distributed multi-player bandits - a game of thrones approach.
\newblock In {\em NeurIPS}, pages 7222--7232, 2018.

\bibitem{BHPQ17}
Avrim Blum, Nika Haghtalab, Ariel~D. Procaccia, and Mingda Qiao.
\newblock Collaborative {PAC} learning.
\newblock In {\em NIPS}, pages 2392--2401, 2017.

\bibitem{BMS09}
S{\'{e}}bastien Bubeck, R{\'{e}}mi Munos, and Gilles Stoltz.
\newblock Pure exploration in multi-armed bandits problems.
\newblock In {\em ALT}, pages 23--37, 2009.

\bibitem{BWV13}
S{\'{e}}bastien Bubeck, Tengyao Wang, and Nitin Viswanathan.
\newblock Multiple identifications in multi-armed bandits.
\newblock In {\em {ICML}}, pages 258--265, 2013.

\bibitem{CL16}
Alexandra Carpentier and Andrea Locatelli.
\newblock Tight (lower) bounds for the fixed budget best arm identification
  bandit problem.
\newblock In {\em {COLT}}, pages 590--604, 2016.

\bibitem{CDS13}
Nicol{\`{o}} Cesa{-}Bianchi, Ofer Dekel, and Ohad Shamir.
\newblock Online learning with switching costs and other adaptive adversaries.
\newblock In {\em NIPS}, pages 1160--1168, 2013.

\bibitem{CGMM16}
Nicol{\`{o}} Cesa{-}Bianchi, Claudio Gentile, Yishay Mansour, and Alberto
  Minora.
\newblock Delay and cooperation in nonstochastic bandits.
\newblock In {\em COLT}, pages 605--622, 2016.

\bibitem{CCDJ17}
Mithun Chakraborty, Kai Yee~Phoebe Chua, Sanmay Das, and Brendan Juba.
\newblock Coordinated versus decentralized exploration in multi-agent
  multi-armed bandits.
\newblock In {\em Proceedings of the Twenty-Sixth International Joint
  Conference on Artificial Intelligence, {IJCAI} 2017, Melbourne, Australia,
  August 19-25, 2017}, pages 164--170, 2017.

\bibitem{CZZ18}
Jiecao Chen, Qin Zhang, and Yuan Zhou.
\newblock Tight bounds for collaborative pac learning via multiplicative
  weights.
\newblock In {\em NIPS}, pages 3602--3611, 2018.

\bibitem{CGL16}
Lijie Chen, Anupam Gupta, and Jian Li.
\newblock Pure exploration of multi-armed bandit under matroid constraints.
\newblock In {\em COLT}, pages 647--669, 2016.

\bibitem{CLQ17}
Lijie Chen, Jian Li, and Mingda Qiao.
\newblock Nearly instance optimal sample complexity bounds for top-k arm
  selection.
\newblock In {\em AISTATS}, pages 101--110, 2017.

\bibitem{CLQ17a}
Lijie Chen, Jian Li, and Mingda Qiao.
\newblock Towards instance optimal bounds for best arm identification.
\newblock In {\em COLT}, pages 535--592, 2017.

\bibitem{CLKLC14}
Shouyuan Chen, Tian Lin, Irwin King, Michael~R. Lyu, and Wei Chen.
\newblock Combinatorial pure exploration of multi-armed bandits.
\newblock In {\em NIPS}, pages 379--387, 2014.

\bibitem{DOR18}
Maria Dimakopoulou, Ian Osband, and Benjamin~Van Roy.
\newblock Scalable coordinated exploration in concurrent reinforcement
  learning.
\newblock In {\em NeurIPS}, pages 4223--4232, 2018.

\bibitem{DR18}
Maria Dimakopoulou and Benjamin~Van Roy.
\newblock Coordinated exploration in concurrent reinforcement learning.
\newblock In {\em ICML}, pages 1270--1278, 2018.

\bibitem{DGW02}
Carlos Domingo, Ricard Gavald{\`{a}}, and Osamu Watanabe.
\newblock Adaptive sampling methods for scaling up knowledge discovery
  algorithms.
\newblock {\em Data Min. Knowl. Discov.}, 6(2):131--152, 2002.

\bibitem{DRY18}
John~C. Duchi, Feng Ruan, and Chulhee Yun.
\newblock Minimax bounds on stochastic batched convex optimization.
\newblock In {\em COLT}, pages 3065--3162, 2018.

\bibitem{EKMM19}
Hossein Esfandiari, Amin Karbasi, Abbas Mehrabian, and Vahab~S. Mirrokni.
\newblock Batched multi-armed bandits with optimal regret.
\newblock {\em CoRR}, abs/1910.04959, 2019.

\bibitem{EMM02}
Eyal Even{-}Dar, Shie Mannor, and Yishay Mansour.
\newblock {PAC} bounds for multi-armed bandit and markov decision processes.
\newblock In {\em {COLT}}, pages 255--270, 2002.

\bibitem{GHRZ19}
Zijun Gao, Yanjun Han, Zhimei Ren, and Zhengqing Zhou.
\newblock Batched multi-armed bandits problem.
\newblock In {\em NeurIPS}, 2019.

\bibitem{GK16}
Aur{\'{e}}lien Garivier and Emilie Kaufmann.
\newblock Optimal best arm identification with fixed confidence.
\newblock In {\em {COLT}}, pages 998--1027, 2016.

\bibitem{GB15}
Zhaohan Guo and Emma Brunskill.
\newblock Concurrent {PAC} {RL}.
\newblock In {\em AAAI}, pages 2624--2630, 2015.

\bibitem{HKKLS13}
Eshcar Hillel, Zohar~Shay Karnin, Tomer Koren, Ronny Lempel, and Oren Somekh.
\newblock Distributed exploration in multi-armed bandits.
\newblock In {\em NIPS}, pages 854--862, 2013.

\bibitem{DPSV12}
Hal~Daum{\'{e}} III, Jeff~M. Phillips, Avishek Saha, and Suresh
  Venkatasubramanian.
\newblock Efficient protocols for distributed classification and optimization.
\newblock In {\em ALT}, pages 154--168, 2012.

\bibitem{Shevtsova10}
Shevtsova Irina.
\newblock An improvement of convergence rate estimates in the lyapunov theorem.
\newblock In {\em Doklady Mathematics}, volume~82, pages 862--864. Springer,
  2010.

\bibitem{JMNB14}
Kevin Jamieson, Matthew Malloy, Robert Nowak, and S{\'e}bastien Bubeck.
\newblock lil?ucb: An optimal exploration algorithm for multi-armed bandits.
\newblock In {\em COLT}, pages 423--439, 2014.

\bibitem{JSXC19}
Tianyuan Jin, Jieming Shi, Xiaokui Xiao, and Enhong Chen.
\newblock Efficient pure exploration in adaptive round model.
\newblock In {\em NeurIPS}, pages 6605--6614, 2019.

\bibitem{JJNZ16}
Kwang{-}Sung Jun, Kevin~G. Jamieson, Robert~D. Nowak, and Xiaojin Zhu.
\newblock Top arm identification in multi-armed bandits with batch arm pulls.
\newblock In {\em AISTATS}, pages 139--148, 2016.

\bibitem{KTAS12}
Shivaram Kalyanakrishnan, Ambuj Tewari, Peter Auer, and Peter Stone.
\newblock Pac subset selection in stochastic multi-armed bandits.
\newblock In {\em ICML}, pages 227--234, 2012.

\bibitem{KLR12}
Varun Kanade, Zhenming Liu, and Bozidar Radunovic.
\newblock Distributed non-stochastic experts.
\newblock In {\em NIPS}, pages 260--268, 2012.

\bibitem{KLMY19}
Daniel Kane, Roi Livni, Shay Moran, and Amir Yehudayoff.
\newblock On communication complexity of classification problems.
\newblock In {\em COLT}, pages 1903--1943, 2019.

\bibitem{KKS13}
Zohar Karnin, Tomer Koren, and Oren Somekh.
\newblock Almost optimal exploration in multi-armed bandits.
\newblock In {\em ICML}, pages 1238--1246, 2013.

\bibitem{KCG16}
Emilie Kaufmann, Olivier Capp{\'{e}}, and Aur{\'{e}}lien Garivier.
\newblock On the complexity of best-arm identification in multi-armed bandit
  models.
\newblock {\em J. Mach. Learn. Res.}, 17:1:1--1:42, 2016.

\bibitem{KL85}
Lloyd~W Koenig and Averill~M Law.
\newblock A procedure for selecting a subset of size m containing the l best of
  k independent normal populations, with applications to simulation.
\newblock {\em Communications in Statistics-Simulation and Computation},
  14(3):719--734, 1985.

\bibitem{LSL16}
Peter Landgren, Vaibhav Srivastava, and Naomi~Ehrich Leonard.
\newblock On distributed cooperative decision-making in multiarmed bandits.
\newblock In {\em ECC}, pages 243--248, 2016.

\bibitem{LV19}
Jasper~CH Lee and Paul Valiant.
\newblock Uncertainty about uncertainty: Near-optimal adaptive algorithms for
  estimating binary mixtures of unknown coins.
\newblock {\em arXiv preprint arXiv:1904.09228}, 2019.

\bibitem{LZ10}
Keqin Liu and Qing Zhao.
\newblock Distributed learning in multi-armed bandit with multiple players.
\newblock {\em {IEEE} Trans. Signal Processing}, 58(11):5667--5681, 2010.

\bibitem{MT04}
Shie Mannor and John~N. Tsitsiklis.
\newblock The sample complexity of exploration in the multi-armed bandit
  problem.
\newblock {\em J. Mach. Learn. Res.}, 5:623--648, 2004.

\bibitem{NZ18}
Huy~L. Nguyen and Lydia Zakynthinou.
\newblock Improved algorithms for collaborative {PAC} learning.
\newblock In {\em NeurIPS}, pages 7642--7650, 2018.

\bibitem{PRCS15}
Vianney Perchet, Philippe Rigollet, Sylvain Chassang, and Erik Snowberg.
\newblock Batched bandit problems.
\newblock In {\em COLT}, page 1456, 2015.

\bibitem{RSS16}
Jonathan Rosenski, Ohad Shamir, and Liran Szlak.
\newblock Multi-player bandits - a musical chairs approach.
\newblock In {\em ICML}, pages 155--163, 2016.

\bibitem{SBC06}
Christian Schmidt, J{\"{u}}rgen Branke, and Stephen~E. Chick.
\newblock Integrating techniques from statistical ranking into evolutionary
  algorithms.
\newblock In {\em Applications of Evolutionary Computing, EvoWorkshops 2006:
  EvoBIO, EvoCOMNET, EvoHOT, EvoIASP, EvoINTERACTION, EvoMUSART, and EvoSTOC},
  pages 752--763, 2006.

\bibitem{SNBWM13}
David Silver, Leonard Newnham, David Barker, Suzanne Weller, and Jason McFall.
\newblock Concurrent reinforcement learning from customer interactions.
\newblock In {\em ICML}, pages 924--932, 2013.

\bibitem{SJR17}
Max Simchowitz, Kevin~G. Jamieson, and Benjamin Recht.
\newblock The simulator: Understanding adaptive sampling in the
  moderate-confidence regime.
\newblock In {\em {COLT}}, pages 1794--1834, 2017.

\bibitem{SB98}
Richard~S. Sutton and Andrew~G. Barto.
\newblock {\em Reinforcement learning - an introduction}.
\newblock Adaptive computation and machine learning. {MIT} Press, 1998.

\bibitem{SBHOJK13}
Bal{\'a}zs Sz{\"o}r{\'e}nyi, R{\'o}bert Busa-Fekete, Istv{\'a}n Heged{\H{u}}s,
  R{\'o}bert Orm{\'a}ndi, M{\'a}rk Jelasity, and Bal{\'a}zs K{\'e}gl.
\newblock Gossip-based distributed stochastic bandit algorithms.
\newblock In {\em ICML}, pages 19--27, 2013.

\bibitem{TZZ19}
Chao Tao, Qin Zhang, and Yuan Zhou.
\newblock Collaborative learning with limited interaction: Tight bounds for
  distributed exploration in multi-armed bandits.
\newblock In {\em FOCS}, pages 126--146, 2019.

\bibitem{Thompson33}
William~R Thompson.
\newblock On the likelihood that one unknown probability exceeds another in
  view of the evidence of two samples.
\newblock {\em Biometrika}, 25(3/4):285--294, 1933.

\bibitem{VWW19}
Santosh~S. Vempala, Ruosong Wang, and David~P. Woodruff.
\newblock The communication complexity of optimization.
\newblock {\em CoRR}, abs/1906.05832, 2019.

\bibitem{XTZV15}
Jie Xu, Cem Tekin, Simpson Zhang, and Mihaela Van Der~Schaar.
\newblock Distributed multi-agent online learning based on global feedback.
\newblock {\em IEEE Transactions on Signal Processing}, 63(9):2225--2238, 2015.

\bibitem{ZDW12}
Yuchen Zhang, John~C. Duchi, and Martin~J. Wainwright.
\newblock Communication-efficient algorithms for statistical optimization.
\newblock In {\em NIPS}, pages 1511--1519, 2012.

\bibitem{ZWSL10}
Martin Zinkevich, Markus Weimer, Alexander~J. Smola, and Lihong Li.
\newblock Parallelized stochastic gradient descent.
\newblock In {\em NIPS}, pages 2595--2603, 2010.

\end{thebibliography}

\appendix

\section{Probability Tools}

\begin{lemma}\label{lem:chernoff}
Let $X_1, \dotsc, X_n \in [0, d]$ be independent random variables and $X = \sum\limits_{i = 1}^n X_i$.
Then 
\[
\Pr[X > \bE[X] + t] \leq \exp\left(-\frac{2 t^2}{n d^2}\right)
\quad
\text{and} 
\quad
\Pr[X < \bE[X] - t] \leq \exp\left(-\frac{2 t^2}{n d^2}\right)
\,.
\]
Moreover, if $X_1, \dotsc, X_n \in [0, 1]$ and $\mu_L \le \bE[X] \le \mu_H$, then we also have for every $\delta \in [0, 1]$, 
\begin{align*}
    \Pr\left[{X}  \geq (1 + \delta) \mu_H \right] \leq  \exp\left({-\frac{\delta^2 \mu_H}{3}}\right)  \quad \text{and} \quad \Pr\left[{X}\le (1-\delta) \mu_L \right] \leq  \exp\left({-\frac{\delta^2 \mu_L}{3}}\right).
\end{align*}
\end{lemma}

\begin{theorem}[Berry-Esseen, \cite{Shevtsova10}]\label{thm:berry-esseen}
Let $X_1, X_2, \dots, X_n$ be independent random variables with $\bE[X_i] = 0$, $\bE[X_i^2] = \sigma_i^2$ and $\bE[|X_i|^3] \leq \rho$ for all $i \in [n]$. Let 
\[
S = \frac{X_1 + X_2 + \dots + X_n}{\sqrt{\sigma_1^2 + \sigma_2^2 + \dots + \sigma_n^2}}.
\]
Let $F$ be the cumulative distribution function of $S$, and $\Phi$ be the cumulative distribution function of the standard normal distribution. It holds that
\[
\sup_{x \in \mathbb{R}} |F(x) - \Phi(x)| \leq 0.5601 \cdot \rho n \left(\sum_{i=1}^{n} \sigma_i^2\right)^{-3/2} .
\]
\end{theorem}

\end{document}